\documentclass[runningheads]{llncs}
\usepackage[letterpaper,margin=1in]{geometry}
\usepackage{graphicx}
\usepackage{booktabs} % For formal tables
\usepackage[ruled]{algorithm2e} % For algorithms
\usepackage{bm}
\usepackage{amsmath,amsfonts}

\def\RP{{\mathbb{R}}_{\geq 0}}
\def\RPP{{\mathbb{R}}_{>0}}

\def\ZP{{\mathbb{N}}}
\def\CG{{\sf LB}}
\def\NCG{{\sf NLB}}
\def\CR{{\sf CR}}
\def\I{{\sf I}}
\def\A{{\sf A}}
\def\N{N}

\def\Z{\mathbb{Z}_{\geq 0}}
\newcommand{\sg}{{\bm\sigma}}
\newcommand{\nsg}{{\bm\Delta}}

\newcommand\Mycite[1]{%
	\cite{#1}}
\newtheorem{fact}{Fact}

\allowdisplaybreaks
\begin{document}
\title{Nash Social Welfare in Selfish and Online Load Balancing\thanks{This work was partially supported by the Italian MIUR PRIN 2017 Project ALGADIMAR “Algorithms, Games, and Digital Markets”.}}
\titlerunning{Nash Social Welfare in Load Balancing}
\author{Vittorio Bilò \inst{1} \and
Gianpiero Monaco \inst{2}  \and
Luca Moscardelli \inst{3}\and 
Cosimo Vinci \inst{4}  
}
\authorrunning{V. Bilò et al.}
\institute{University of Salento, Italy  \and University of L'Aquila, Italy \and University of Chieti-Pescara, Italy \\ \and Gran Sasso Science Institute, Italy} 
%\email{vittorio.bilo@unisalento.it, gianpiero.monaco@univaq.it, luca.moscardelli@unich.it, cosimo.vinci@gssi.it}
\maketitle              % typeset the header of the contribution
\begin{abstract}
In load balancing problems there is a set of clients, each wishing to select a resource from a set of permissible ones, in order to execute a certain task. Each resource has a latency function, which depends on its workload, and a client's cost is the completion time of her chosen resource. Two fundamental variants of load balancing problems are {\em selfish load balancing} (aka. {\em load balancing games}), where clients are non-cooperative selfish players aimed at minimizing their own cost solely, and {\em online load balancing}, where clients appear online and have to be irrevocably assigned to a resource without any knowledge about future requests.
We revisit both selfish and online load balancing under the objective of minimizing the {\em Nash Social Welfare}, i.e., the geometric mean of the clients' costs. To the best of our knowledge, despite being a celebrated welfare estimator in many social contexts, the Nash Social Welfare has not been considered so far as a benchmarking quality measure in load balancing problems. We provide tight bounds on the price of anarchy of pure Nash equilibria and on the competitive ratio of the greedy algorithm under very general latency functions, including polynomial ones. For this particular class, we also prove that the greedy strategy is optimal as it matches the performance of any possible online algorithm.
\end{abstract}
\section{Introduction}

In load balancing problems there is a set of clients, each wishing to select a resource from a set of permissible ones, in order to execute a certain task. Each resource has a latency function, which depends on its workload, and a client's cost is the completion time of her chosen resource. These problems stand at the foundations of the Theory of Computing and have been studied under a variety of objective functions, such as the maximum client's cost (aka. the makespan) \cite{Grah,HS87,HS76,LST90} and the average weighted client's cost (see \cite{CK04} for an excellent survey). 

Two extensively studied variants of load balancing problems are {\em selfish load balancing} \cite{V07} (aka. {\em load balancing games}) and {\em online load balancing} \cite{Grah}.
Selfish load balancing, where clients are non-cooperative selfish players aimed at minimizing their own cost solely, constitutes a notable subclass of \emph{weighted congestion games} \cite{R73} and, as such, enjoys some nice theoretical properties. For instance, they always admit pure Nash Equilibria \cite{I+05}. Moreover, under the assumption that all tasks have unitary weight ({\em unweighted congestion games}), any best-response dynamics converges to a pure Nash Equilibrium in polynomial time \cite{HRV08}. In online load balancing, instead, clients appear online and have to be irrevocably assigned to a resource without any knowledge about future requests.

%Let us first focus on the simpler class of unweighted load balancing problems. 
Interpreting the set of clients of a load balancing problem as a society and adopting the terminology of welfare economics, the makespan and the average weighted client's cost objective functions get called, respectively, the {\em egalitarian} and the {\em utilitarian} social function.
In the case of unweighted tasks, the egalitarian function is defined as $\max_{i}x_i$, and the utilitarian one is defined as $\frac{1}{n}\sum_i x_i$, where $n$ is the number of clients and ${\bm x}=(x_1,x_2,...)$ is the vector encoding the clients' costs. Another interesting social function is the {\em Nash Social Welfare} (NSW) \cite{Nash50}, which is defined as $\left(\prod_i x_i\right)^{\frac 1 n}$, i.e., as the geometric mean of the clients' costs. These definitions naturally extend to the more general case of weighted tasks (see Section \ref{sec_model}).

The NSW is a celebrated welfare measure in many settings, such as Fisher markets \cite{B+10,BS00} and fair division \cite{AW19,BGM17,BMSY19,CGH19,CKMP016,CG18,GHM18}, as it satisfies a set of interesting properties (most of the aforementioned papers focus on fairness properties such as envy-freeness and maximin share, and Pareto optimality) and achieves a balanced compromise between the equity of the egalitarian social welfare function and the efficiency of the utilitarian one. We notice that when $x_i>0$, for any $i=1,\ldots,n$, this balance holds regardless of whether the objective is maximizing or minimizing the NSW. The case where each $x_i$ can be either a positive or negative value has been considered in \cite{AW19}. In the context of congestion games we do not take into account envy-freeness and maximin share, however, it is easy to see that an outcome that minimizes the NSW is Pareto optimal. Another interesting motivation for considering the NSW in load balancing comes from the following observation. An alternative reasonable way to define a client's cost can come by taking the ratio between the completion time of her chosen resource and the completion time she could obtain when being the only client in the system (i.e., when she is the unique user of the fastest resource). This definition avoids situations where the cost of a specific client determines almost completely the value of the social welfare. This happens, for instance, when there is a client $i$ owing a highly time-consuming task. Here, both the utilitarian and the egalitarian social welfare end up depending on the cost of $i$, thus almost neglecting the other clients' costs. In this setting, the NSW is the proper metric to use. More generally, the NSW is the only correct mean to use when averaging normalized results, that is, results that are presented as ratios to reference values \cite{FW86}. It is important to emphasize the scale-freeness of the NSW in load balancing problems, that is, the NSW is a robust social welfare function as its analysis is not affected by this change in the definition of a client's cost. 

%We revisit both selfish and online load balancing under the objective of minimizing the NSW. To the best of our knowledge, this is the first work adopting the NSW as a benchmarking quality measure in resource allocation problems. We provide tight bounds on the price of anarchy \cite{KP99} of pure Nash equilibria (the loss in optimality due to selfish behavior) and on the competitive ratio of the greedy algorithm (the loss in optimality due to lack of information) under very general latency functions, including polynomial ones. These questions have been widely addressed under the utilitarian and egalitarian functions, but never under the NSW.

\iffalse

{\bf Riguardo il problema di ottimizzazione centralizzato, farei una discussione a parte, per esempio nelle conclusioni.} Furthermore, it is worth noticing that, on the one hand, when considering unweighted players, the computation of the optimal configuration with respect to NSW can be trivially done in polynomial time by exploiting the same techniques developed in \cite{CMNV05,MS12} for the utilitarian social welfare (\cite{CMNV05,MS12} use an approach similar to the one adopted in \cite{FPT04} for the computation of a Nash equilibrium); on the other hand, when considering weighted players, a simple reduction from the NP-complete problem $\mathsf{PARTITION}$ shows that the problem becomes NP-hard.

\fi

\subsection{Related Work}
\iffalse
{\color{magenta}
[Luca] Load balancing games:
\begin{itemize}
\item PoA e PoS
\item Online
\item one-round (opzionale)
\item tasse
\item calcolo ottimo e equilibrio/equilibrio migliore
\end{itemize}

\noindent [Gianpiero] Nash social welfare:
\begin{itemize}
\item origini e paper che ne sottolinenano l'importanza
\item setting in cui è stato recentemente utilizzato: allocazione di item (STOC, EC, paper Ioannis EC), ...
\end{itemize}
}
\fi

%{\color{teal}
\noindent\textbf{Selfish Load Balancing.} The literature concerning the efficiency of Nash equilibria in selfish load balancing is highly tied with that of its superclass of congestion games. In the following, we first focus on results for the mostly studied case of the \emph{utilitarian social welfare}. In this setting, it is assumed that all clients selecting the same resource experience the same cost.

The efficiency of pure Nash equilibria in congestion games has been first considered in \Mycite{AAE05} and \Mycite{CK05}, where it has been independently shown that the price of anarchy is $5/2$ and $(3+\sqrt{5})/2$ for, respectively, unweighted and weighted congestion games with affine latency functions. These bounds have been extended to load balancing games in \Mycite{CFKKM11}. However, under the additional assumption that the game is symmetric (i.e., all resources are available to any client), the price of anarchy improves to $4/3$ \Mycite{LMMR08}. Exact bounds for both weighted and unweighted congestion games with polynomial latency functions have been given in \Mycite{ADGMS11}, and \Mycite{GS07,BGR10} prove that they hold even for unweighted load balancing games and symmetric weighted load balancing games, respectively. These results have been further generalized in \Mycite{BV17}, where it is proved that, under general latency functions encompassing polynomial ones, the worst-case price of anarchy of both symmetric weighted congestion games and unweighted congestion games is attained by load balancing instances. This worst-case behavior, however, does not occur under identical resources, where load balancing games exhibit better performance with respect to general congestion games. For instance, for affine latency functions, the price of anarchy drops to $2.012067$ for unweighted games \Mycite{CFKKM11,STZ07} and to $9/8$ for symmetric weighted games \Mycite{LMMR08}. Tight bounds for this last class of games under polynomial and more general latency functions have been given in \cite{GLMM06,BV17}.

For the class of non-atomic congestion games (a variant assuming that each client's task is infinitesimally small with respect to the workload required by the whole society and suited to model communication and transportation networks) \Mycite{R03,roughowb,rougboun} provide bounds on the price of anarchy under general latency functions and prove that they are tight even for a two-node network with two parallel links. An interesting connection between load balancing games and non-atomic congestion games has been uncovered in \Mycite{F10} where it is shown that, under fairly general latency functions, the price of anarchy of unweighted symmetric load balancing games coincides with that of non-atomic congestion games.

Less has been done for the \emph{egalitarian social welfare}. The study of the price of anarchy was initiated in \cite{KP99}, where weighted congestion games of $m$ parallel links with linear latency functions are considered. The price of anarchy for the egalitarian social welfare is $\Theta(\frac{\log m}{\log \log m})$. The lower bound was shown in \cite{KP99} and the upper bound in \cite{czum}. For load balancing games, the price of anarchy is $\Theta(\frac{\log n}{\log \log n})$ where $n$ is the number of players \cite{GLMM06}, while for unweighted congestion games is $\Theta(\sqrt{n})$ \cite{CK05}. \cite{R04} proves that the price of anarchy of non-atomic congestion games with general non-decreasing latency function is $\Omega(n)$. %Some results on the price of stability for the egalitarian social welfare can be found in \cite{CK05b}.

\smallskip
\noindent\textbf{Online Load Balancing.} The performance of greedy load balancing with respect to the utilitarian social welfare and under affine latency functions has been studied in \cite{AAG+95,CFKKM11,STZ07}. \cite{AAG+95} considers a more general model where each client has a load vector denoting her impact on each resource (i.e., how much her assignment to a resource will increase its load) and the objective is to minimize the $L_p$ norm of the load of the resources. Their results, together with \cite{CFKKM11}, imply a competitive ratio of the greedy algorithm equal to $3+2\sqrt{2}\approx 5.8284$ for the utilitarian social welfare. 
This bound carries over also to the case of weighted clients where the objective is to minimize the weighted average latency. 
\cite{STZ07} and \cite{CFKKM11} provide a tight  bound of $17/3$ for different resources and show that
the competitiveness of greedy load balancing is between $4$ and
$\frac{2}{3}\sqrt{21}+1\approx 4.05505$ for identical resources. 
\cite{BV17} characterizes the competitive ratio of the greedy algorithm applied to congestion games with general latency functions.

\cite{BFFM09,CMS06} analyse a different online algorithm (usually termed one-round walk starting from the empty state) for load balancing and prove that its competitive ratio is $2+\sqrt{5}$ under affine latency functions. Bounds for the case of polynomial latencies are given in \cite{B18,BV19,KST19}, while \cite{BV17,V19} address more general latency functions with respect to atomic and non-atomic congestion games, respectively.

Concerning the egalitarian social welfare, most of the results of the literature investigate the case of identical resources, usually termed as machines. \cite{A99,BFKV95,FKT89,FW00,GW93,Grah,KPT96}. We notice that the machine scheduling problem with related (resp. identical) machines is a special case of our weighted load balancing problem with linear latency functions (resp. identical resources with linear latency functions). For $m$ identical machines, \cite{Grah} shows that the greedy algorithm achieves a competitive ratio of exactly $2 - \frac{1}{m}$ and this bound is proven the best possible one for $m=2,3$ in \cite{FKT89}. The currently best known algorithm achieves a competitive ratio of $1.9201$ \cite{FW00} for any $m$ and no algorithm can achieve a competitive ratio bettern than $1.88$ \cite{R01}. For related machines, \cite{AAFPW97,ANR92} show a tight bound of $\log m$, while \cite{C08} considers the case of unrelated machines with the objective of minimizing the norm of the machines loads.

%}

%\input{intro_RW_NSW}

\subsection{Our Contribution}

\iffalse
{\color{magenta}
\begin{itemize}
\item PoA LBG weighted (tight UB e LB under mild assumptions)
\item PoA LBG unweighted (tight UB e LB under mild assumptions)
\item PoA non-atomic (tight UB e LB under mild assumptions)
%\item Calcolo ottimo e equilibrio migliore: poly per unweighted (reduction to min-cost flow), NP-HARD per weighted (reduction from PARTITION).
\item Competitive ratio (UB dato dal greedy e due LB tight: per funzioni generali all'algoritmo greedy e per funzioni polinomiali al problema)
\end{itemize}
}
\fi
%We revisit both selfish and online load balancing under the objective of minimizing the {\em Nash Social Welfare}, i.e., the geometric mean of the clients' costs. To the best of our knowledge, despite being a celebrated welfare estimator in many social contexts, the Nash Social Welfare has not been considered so far as a benchmarking quality measure in resource allocation problems. 

We revisit both selfish and online load balancing under the objective of minimizing the NSW. To the best of our knowledge, this is the first work adopting the NSW as a benchmarking quality measure in load balancing problems.   
We analyze the price of anarchy \cite{KP99} of pure Nash equilibria (the loss in optimality due to selfish behavior) and the competitive ratio of online algorithms (the loss in optimality due to lack of information) under very general latency functions. These questions have been widely addressed under the utilitarian and egalitarian functions, but never under the NSW.
%We analyze the performance of Load Balancing under the objective of minimizing the {\em Nash Social Welfare} in two main settings, namely the selfish setting and the online one, by providing upper and lower bounds to the price of anarchy and to the competitive ratio. 

%%1
%dire che  facendo il rapporto di 2 stati si ottiene il rapporto di 2 log e non il log del rapporto come sarebbe desiderabile per ottrenere il risultato

We notice that by adopting the NSW as new metric, we are not going to modify the set of Nash equilibria but only the social values. The main difference between the NSW and the classical notion of utilitarian social welfare consists in the fact that, while in the latter the players' costs are summed, in the former they are multiplied.
% (and finally the $n$-root is applied, with $n$ being the number of players). 
This may lead to think that, by turning the costs into their logarithms, a classical utilitarian analysis can be easily adapted to deal with the NSW. % for applying . 
Actually, this is not the case. In fact, on the one hand, using this idea for bounding a performance ratio (e.g., the price of anarchy or the competitive ratio), one obtains a bound on the ratio between two logarithms (each one having the product of the players' costs as argument). On the other hand, we are interested in bounding the ratio between the argument of these logarithms, and there is no direct correlation between these two ratios (notice that logarithm of the latter ratio is equal to the difference between the corresponding utilitarian social costs, and therefore it is not related to the former one). Thus, the analysis of the NSW requires different proof arguments.   
%Actually, this is not the case because, when considering the logarithm of the ratio, we obtain the difference between the sum of the logarithms of the individual costs at the equilibrium, and the sum of the individual costs at the optimum, both divided by the number of players. These sums can be seen as utilitarian social costs with respect to the logarithms of the initial individual costs. Anyway, the fact that we are evaluating the difference between these utilitarian social costs (divided by the number of players), and not the ratio, yields to different results and requires different proof arguments. 
%%2
%Dire che il NSW è diverso dall'utilitarian proprio a livello di risultati ottenuti in quanto nella somma abbiamo che l'UB del congestion generale è matchato dal LB del load balancing, mentre nel NSW abbiamo in appendice che L'UB per il congestion non LB even con funzioni lineari è ordine di n mentre si abbassa a costante considerando il LB.
In order to have another evidence of this fact, it is worth noticing that the results obtained for the NSW substantially differ from the ones holding for the utilitarian social function, not only from a quantitative point of view, but also from a qualitative one. In fact, while it is well known (see \Mycite{CFKKM11}) that for the utilitarian social welfare the simpler combinatorial structure of load balancing games does not improve the price of anarchy of general congestion games, our Theorem \ref{thm_CG} (deferred to the appendix) and Corollary \ref{cor1} show that, for the NSW, even for the case of linear latency functions, the price of anarchy drops from $n$ to $2$.

All upper bounds shown in this paper are quite general, given that they hold for any non-decreasing and positive latency function. Moreover, the provided matching lower bounds hold for latency functions verifying mild assumptions; it is worth to remark that they are satisfied by the well studied class of polynomial latency functions and by many other ones.

In particular, Theorem \ref{thm_w_upp} provides an upper bound to the price of anarchy for the case of weighted load balancing games, while Theorem \ref{thm_w_low} gives a matching lower bound. 
Similarly, we focus on unweighted games (a special case of weighted ones) by providing tight bounds that, in general, are lower than the ones that can be obtained for weighted games (see Subsection \ref{subsec_nash_unweighted}). 
However, Corollaries \ref{cor1} (or \ref{cor1b}) and \ref{cor2} show that, when considering polynomial latency functions of degree $p$, the two analyses (for weighted games and for unweighted ones) give the same tight bound of $2^p$. Furthermore, when considering weighted games, the tight bound of $2^p$ holds even for symmetric games (Corollary \ref{cor1}) and for games with identical resources (Corollary \ref{cor1b}). We also provide a tight analysis holding for non-atomic games (see Subsection \ref{subsec_nash_nonatomic}); for the case of polynomial latency functions of degree $p$, Corollary \ref{cor3} shows that the price of anarchy is $\left( e ^{\frac 1 e} \right) ^ p \simeq (1.44)^p$. For the online setting, we analyze the greedy algorithm that assigns every client to a resource minimizing the total cost of the instance revealed up to the time of its appearance. We provide a tight analysis of the  competitive ratio of the greedy algorithm, and we show that, when considering polynomial latency functions of degree $p$, there exists no online algorithm achieving a competitive ratio better than the one of the greedy algorithm, that is equal to $4^p$ (see Section \ref{sec_onlineLB}). In Table \ref{tab:1}, we consider the case of polynomial latency functions, and we compare the performance under the NSW with that under the utilitarian social welfare studied in some previous works. 

The rest of the paper is structured as follows. 
Section \ref{sec_model} introduces the model. 
Sections \ref{sec_selfishLB} and \ref{sec_onlineLB} are devoted to the performance analysis of the price of anarchy and of the competitive ratio, under the selfish and the online setting, respectively. 
Finally, in Section \ref{sec_conlusion} we give some conclusive remarks and state some interesting open problems. Due to lack of space, some proofs are sketched or omitted, and are left to the appendix.
%\begin{table}
%\begin{tabular}{| c || c |}
%\hline
% & NSW of Weighted (W), Unweighted (U), Non-atomic (N), Online (O)\\ 
%\hline\hline
%W & $\sup_{k_1\geq o_1>  0,o_2>k_2\geq 0,f_1,f_2\in \mathcal{C}}\left(\frac{f_1(k_1+o_1)}{f_1(o_1)}\right)^{\frac{(o_2-k_2)o_1}{k_1 o_2-k_2 o_1}}\left(\frac{f_2(k_2+o_2)}{f_2(o_2)}\right)^{\frac{(k_1-o_1)o_2}{k_1 o_2-k_2 o_1}}$ \\ \hline 
%U & $\sup_{f\in\mathcal{C}, k\in \ZP, o\in  \{1,\ldots, k\}}\left(\frac{f(k+1)}{f(o)}\right)^{\frac{o}{k}}$\\ \hline
%N & $\sup_{f\in\mathcal{C},  k\geq o>0}\left(\frac{f(k)}{f(o)}\right)^{\frac{o}{k}}$ \\ \hline
%O & $\sup_{k_1\geq  o_1> 0,o_2>k_2\geq 0,f_1,f_2\in \mathcal{C}}\left(\frac{f_1(k_1+o_1)^{k_1+o_1}}{f_1(k_1)^{k_1}f_1(o_1)^{o_1}}\right)^{\frac{o_2-k_2}{o_2 k_1-o_1k_2}}\left(\frac{f_2(k_2+o_2)^{k_2+o_2}}{f_2(k_2)^{k_2}f_2(o_2)^{o_2}}\right)^{\frac{k_1-o_1}{o_2 k_1-o_1k_2}}$ \\ \hline
%\end{tabular}
%\caption{\label{tab:0}Tight bounds on the performance of load balancing with latency functions in $\mathcal{C}$, where $\mathcal{C}$ is a class of general latency functions verifying some mild assumptions.}
%\end{table}
\begin{table}
\begin{center}
\begin{tabular}{| c || c | c |}
\hline
 & NSW & USW \\ 
\hline\hline
Weighted & $2^p$ & $(\Phi_p)^{p+1}\sim \Theta\left(\frac{p}{\log(p)}\right)^{p+1}$,  \cite{ADGMS11}\\ \hline 
Unweighted & $2^p$ & $\frac{(k+1)^{2p+1}-k^{p+1}(k+2)^p}{(k+1)^{p+1}-(k+2)^p+(k+1)^p-k^{p+1}}\sim \Theta\left(\frac{p}{\log(p)}\right)^{p+1}$, \cite{ADGMS11}\\ \hline
Non-atomic & $\left(e^{\frac{1}{e}}\right)^p$ & $\left(1-p(p+1)^{-(p+1)/p}\right)^{-1}\sim \Theta\left(\frac{p}{\log(p)}\right)$, \cite{R03}\\ \hline
Online & $4^p$ & $(2^{1/(p+1)}-1)^{-(p+1)}\sim \Theta(p)^{p+1}$, \cite{C08}\\ \hline
\end{tabular}
\caption{\label{tab:1}Tight bounds on the performance of load balancing with polynomial latency functions of maximum degree $p$, under the NSW and the utilitarian social welfare (USW). $\Phi_p$ denotes the unique solution of equation $x^{p+1}=(x+1)^p$, and $k:=\lfloor \Phi_p\rfloor$. We observe that the performance under the NSW case is definitely better (even asymptotically) than that under the USW case, except for the non-atomic setting.}
\end{center}
\end{table}
\section{Model}\label{sec_model}
Given $k\in \mathbb{N}$, let $[k]:=\{1,2,\ldots, k\}$.
A class $\mathcal{C}$ of functions is called 
%\setlist{nolistsep}
%\begin{itemize}[noitemsep]
%\itemsep0em
%\item 
{\em ordinate-scaling} if, for any $f\in\mathcal{C}$ and $\alpha\geq 0$, the function $g$ such that $g(x)=\alpha f(x)$ for any $x\geq 0$, belongs to $\mathcal{C}$;
%\item
{\em abscissa-scaling} if, for any $f\in\mathcal{C}$ and $\alpha\geq 0$, the function $g$ such that $g(x)=f(\alpha x)$ for any $x\geq 0$, belongs to $\mathcal{C}$;
%\item
{\em all-constant-including} if it contains all the constant functions (i.e., all functions $f$ such that $f(x)=c$ for some $c>0$);
%\item
{\em unbounded-including} if all the latency functions $f$, except for the constant ones, verify $\lim_{x\rightarrow \infty}f(x)=\infty$.
%\end{itemize}
Let $\mathcal{P}(p)$ denote the class of polynomial latencies of maximum degree $p$, i.e., the class of functions $f(x)=\sum_{d=0}^p\alpha_{d} x^d$, with $\alpha_{d}\geq 0$ for any $d\in [p]\cup\{0\}$ and $\alpha_{d}>0$ for some $d\in [p]\cup\{0\}$. A function $f$ is {\em quasi-log-convex} if $x\ln(f(x))$ is convex. 

We first deal with \emph{selfish load balancing}, by defining \emph{load balancing games}, and then we turn our attention to the online setting.

\subsection{Selfish Load Balancing}

\noindent\textbf{(Atomic) Load balancing games.} A {\em weighted (atomic) load balancing game}, or {\em load balancing game} for brevity, is a tuple $\CG=\left(\N,R,(\ell_j)_{j\in R},(w_i)_{i\in\N},({\Sigma}_i)_{i\in\N}\right),$ where $\N$ is a set of $n\geq 1$ players (corresponding to clients), $R$ is a finite set of resources, $\ell_j:\RPP\rightarrow\RPP$ is the (non-decreasing and positive) latency function of resource $j\in R$, and, for each $i\in\N$, $w_i>0$ is the weight of player $i$ and ${\Sigma}_i\subseteq R$ (with ${\Sigma}_i\neq\emptyset$) is her set of strategies (or admissible resources). For notational simplicity, we assume that each latency function $\ell$ verifies $\ell(0)=0$.

An {\em unweighted load balancing game} is a weighted load balancing game with unitary weights. A {\em symmetric weighted load balancing game} is a congestion game in which each player can select all the resources, i.e., ${ \Sigma}_i=R$ for any $i\in \N$. 
%A {\em weighted load balancing game with identical resources} is a weighted load balancing game such that each resource has the same latency function, i.e. $\ell_j:=\ell$ for any $j\in R$, for some latency function $\ell$.

%In the remainder of the paper, we will only consider load balancing games. In particular, 
Given a class $\mathcal{C}$ of latency functions, 
%let ${\sf W}(\mathcal{C})$ be the class of weighted congestion games, 
%${\sf U}(\mathcal{C})$ be the class of unweighted congestion games, 
let ${\sf ULB}(\mathcal{C})$ be the class of unweighted load balancing games, ${\sf WLB}(\mathcal{C})$ be the class of weighted load balancing games, and ${\sf SWLB}(\mathcal{C})$ be the class of weighted symmetric load balancing games, all having latency functions in the class $\mathcal{C}$. We say that resources are {\em identical} if all of them have the same latency function.
%Given $n\in\mathbb{N}$ (resp. $x>0$), let ${\sf U}_n(\mathcal{C})$ (resp. ${\sf W}_x(\mathcal{C})$, be the class of unweighted (resp. weighted) congestion games with latency functions in $\mathcal{C}$ and at most $n$ players (resp. with total players' weight at most $x$). 

\smallskip
\noindent\textbf{Non-Atomic Load balancing Games.} The counterpart of the class of atomic load balancing games is that of {\em non-atomic load balancing games} \Mycite{Beckmann1956,pigo,ward}: these games are a good approximation for atomic ones when players become infinitely many and the contribution of each player to social welfare becomes infinitesimally small. 
A {\em non-atomic load balancing game} is a tuple $\NCG=\left(\N,R,(\ell_j)_{j\in R},(r_i)_{i\in \N},({\Sigma}_i)_{i\in\N}\right)$, where $\N$ is a set of $n\geq 1$ \emph{types} of players, $R$ is a finite set of resources, $\ell_j:\RPP\rightarrow\RPP$ is the (non-decreasing and positive) latency function of resource $j\in R$; moreover, given $i\in\N$, $r_i\in \RP$ is the amount of players of type $i$ and ${\Sigma}_i \subseteq R$ is the set of strategies of every player of type $i$. %, that is, each player of type $i$ has $m_i\geq 1$ possible choices.

Given a class $\mathcal{C}$ of latency functions, 
%let ${\sf W}(\mathcal{C})$ be the class of weighted congestion games, 
%${\sf U}(\mathcal{C})$ be the class of unweighted congestion games, 
let ${\sf NLB}(\mathcal{C})$ be the class of non-atomic load balancing games, and ${\sf SNLB}(\mathcal{C})$ be the class of symmetric non-atomic load balancing games, all having latency functions in the class $\mathcal{C}$. 

\smallskip
\noindent\textbf{Strategy Profiles and Cost Functions.} In atomic load balancing games, a {\em strategy profile} is an $n$-tuple ${\bm\sigma}=(\sigma_1,\ldots,\sigma_n)$, where $\sigma_i\in {\Sigma}_i$ is the resource chosen by each player $i\in \N$ in $\sg$. Given a strategy profile $\sg$, let $k_j(\sg):=\sum_{i\in \N:\sigma_i=j}w_i$ be the {\em congestion} of resource $j\in R$ in $\sg$, and let 
%$cost_i(\sg):=\sum_{j\in \sigma_i}\ell_j(k_j(\sg))$ 
$cost_i(\sg):=\ell_{\sigma_i}(k_{\sigma_i}(\sg))$  be the {\em cost} of player $i\in \N$ in $\sg$. 

In non-atomic load balancing games, a {\em strategy profile} is an $n$-tuple  $\nsg=(\Delta_1,\ldots,\Delta_n)$, where $\Delta_i:{\Sigma}_i\rightarrow \RP$ is a function denoting, for each resource $j\in {\Sigma}_i$, the amount $\Delta_i(j)$ of players of type $i$ selecting resource $j$, so that $\sum_{j\in {\Sigma}_i}\Delta_i(j)=r_i$. 
Observe that $\Delta_i(j)=0$ if $j\notin{\Sigma}_i$.  
For a strategy profile $\bm\Delta$, the congestion of resource $j\in R$ in $\bm\Delta$, denoted as $k_j(\nsg):=\sum_{i\in\N}\Delta_i(j)$, is the total amount of players using resource $j$ in $\bm\Delta$ and its cost is given by $cost_{j}(\nsg)=\ell_j(k_j(\nsg))$.
The cost of a player of type $i$ selecting a resource $j \in {\Sigma}_i$ is equal to  $cost_{j}(\nsg)$ and each player aims at minimizing it.

\smallskip
\noindent\textbf{Nash Social Welfare.} In atomic load balancing games, the {\em  Nash Social Welfare (NSW)} of a strategy profile $\sg$ is defined as:
$
{\sf NSW}(\sg):=\left({\prod_{i\in \N}cost_i(\bm \sigma)^{w_i}}\right)^{\frac{1}{\sum_{i\in \N}w_i}}.
$
Using the previous definition, for unweighted games we get ${\sf NSW}(\sg)=\left({\prod_{i\in \N}cost_i(\bm\sigma)}\right)^{\frac{1}{n}}$. 
Given a strategy profile $\sg$, let $R(\sg):=\{j\in R: k_j(\sg)>0\}$. For weighted load balancing games we get:
${\sf NSW}(\sg)=\left({\prod_{i\in \N}cost_i(\sg)}\right)^{\frac{1}{\sum_{i\in \N}w_i}}=\left({\prod_{j\in R(\sg)}\ell_j(k_{j}(\sg))^{k_{j}(\sg)}}\right)^{\frac{1}{\sum_{i\in \N}w_i}}=\left({\prod_{j\in R(\sg)}\ell_j(k_{j}(\sg))^{k_{j}(\sg)}}\right)^{\frac{1}{\sum_{j\in R(\sg)}k_j(\sg)}}$.
\iffalse
\begin{align*}
{\sf NSW}(\sg)&=\left({\prod_{i\in \N}cost_i(\sg)^{w_i}}\right)^{\frac{1}{\sum_{i\in \N}w_i}}=\left({\prod_{j\in R(\sg)}\ell_j(k_{j}(\sg))^{k_{j}(\sg)}}\right)^{\frac{1}{\sum_{i\in \N}w_i}}\\
&=\left({\prod_{j\in R(\sg)}\ell_j(k_{j}(\sg))^{k_{j}(\sg)}}\right)^{\frac{1}{\sum_{j\in R(\sg)}k_j(\sg)}}.
\end{align*}
\fi

Let ${\sf SP}(\CG)$ be the set of strategy profiles of an atomic load balancing game $\CG$. An optimal strategy profile $\sg^*(\CG)$ of a load balancing game $\CG$ is a strategy profile $\sg^*\in \arg\min_{\sg\in {\sf SP}(\CG)}{\sf NSW}(\sg)$, i.e., a strategy profile minimizing the NSW.

Analogously, for the non-atomic setting, we have
${\sf NSW}(\nsg)=\left({\prod_{j\in R(\nsg)}cost_j(\nsg)^{k_{j}(\nsg)}}\right)^{\frac{1}{\sum_{j\in R(\nsg)}k_j(\nsg)}},$
where $R(\nsg):=\{j\in R: k_j(\nsg)>0\}$.
Let ${\sf SP}(\NCG)$ be the set of strategy profiles of a non-atomic load balancing game $\NCG$. An optimal strategy profile $\nsg^*(\NCG)$ of a load balancing game $\NCG$ is a strategy profile $\nsg^*\in \arg\min_{\nsg\in {\sf SP}(\NCG)}{\sf NSW}(\nsg)$, i.e., a strategy profile minimizing the NSW.

\smallskip
\noindent\textbf{Pure Nash Equilibria and their Efficiency.} In the atomic setting, for a given strategy profile $\sg$, let $(\sg_{-i},\sigma_i'):=(\sigma_1,\sigma_2,\ldots, \sigma_{i-1},\sigma_i',\sigma_{i+1},\ldots, \sigma_n)$, i.e., a strategy profile equal to $\sg$, except for strategy $\sigma_i'$. A {\em pure Nash equilibrium} is a strategy profile $\sg$ such that $cost_i(\sg)\leq cost_i(\sg_{-i},\sigma_i')$ for any $\sigma_i'\in {\Sigma}_i$ and $i\in \N$, i.e., a strategy profile in which no player can improve her cost by unilateral deviations. 
%\paragraph*{Nash Price of Anarchy.}
Let ${\sf PNE}(\CG)$ be the set of pure Nash equilibria of a load balancing game $\CG$. The {\em Nash price of anarchy} of $\CG$ is defined as:
$
{\sf NPoA}(\CG)=\sup_{\sg\in {\sf PNE}(\CG)}\frac{{\sf NSW}(\sg)}{{\sf NSW}(\sg^*(\CG))}
$
Given a class $\mathcal{G}$ of load balancing games, the {\em Nash price of anarchy} of $\mathcal{G}$ is defined as ${\sf NPoA}(\mathcal{G})=\sup_{\CG\in\mathcal{G}}{\sf NPoA}(\CG)$. In the non-atomic setting, a {\em pure Nash equilibrium} is a strategy profile $\nsg$ such that, for any player type $i\in\N$, resources $j, j'\in{\Sigma}_i$ such that $\Delta_i(j)>0$, $cost_{j}(\nsg)\leq cost_{j'}(\nsg)$ holds, that is, an outcome of the game in which no player can improve her situation by unilaterally deviating to another strategy. The {\em Nash price of anarchy} of a non-atomic game $\NCG$ (denoted as ${\sf NPoA}(\NCG)$) is defined as in the atomic setting, and again, given a class $\mathcal{G}$ of non-atomic load balancing games, the {\em Nash price of anarchy} of $\mathcal{G}$ is defined as ${\sf NPoA}(\mathcal{G})=\sup_{\NCG\in\mathcal{G}}{\sf NPoA}(\NCG)$.

\subsection{Online Load Balancing}
We now introduce online load balancing. There is a natural correspondence between a load balancing game and an instance of the online load balancing problem. When dealing with the online setting, as usual in the literature, we adopt a different nomenclature. In particular, an instance $\I$ of the online load balancing problem is a tuple $\I=\left(\N,R,(\ell_j)_{j\in R},(w_i)_{i\in\N},({\Sigma}_i)_{i\in\N}\right),$ where $\N=[n]$ is a set of $n\geq 1$ \emph{clients}, $R$ is a finite set of resources, $\ell_j:\RPP\rightarrow\RPP$ is the (non-decreasing and positive) latency function of resource $j\in R$, and, for each $i\in\N$, $w_i>0$ is the weight of client $i$ and ${\Sigma}_i\subseteq R$ (with ${\Sigma}_i\neq\emptyset$) is her set of \emph{admissible resources}. Furthermore, in the online setting an assignment of clients to resources is called state: A {\em state} is an $n$-tuple ${\bm\sigma}=(\sigma_1,\ldots,\sigma_n)$, where $\sigma_i\in {\Sigma}_i \subseteq R$ is the resource assigned to player $i\in \N$ in $\sg$. As in load balancing games, given a class of latency latency functions $\mathcal{C}$, let ${\sf WLB}(\mathcal{C})$ denote class of load balancing instances with latency functions in $\mathcal{C}$. 

The NSW of a state and the optimal state are defined analogously to the selfish load balancing setting.

\smallskip
\noindent\textbf{The online setting.} In \emph{online load balancing}, clients appear in online fashion, in consecutive \emph{steps}; when a client appears, an irrevocable decision has to be taken in order to assign it to a resource. We assume w.l.o.g. that clients appear in increasing order, i.e., client $i\in [n]$ appears before client $j \in [n]$ if and only if $i<j$.
More formally, for any $i \in [n]$, an online algorithm has to assign client $i$ to a resource being admissible for it without the knowledge of the future clients $i+1,i+2,\ldots$; the assignment of client $i$ decided by the algorithm at step $i$ cannot be modified at later steps.

Notice that at each step $i>1$ a new instance is obtained by adding client $i$ to the instance of step $i-1$.
%In our model, servers have linear latency functions and the objective is to minimize the total latency, i.e., the sum of the latencies experienced by all clients. Clients may also own jobs with non-negative weights; in this case, the objective is to minimize the weighted sum of the latencies experienced by all clients.

\smallskip
\noindent\textbf{Competitive Ratio.} Following the standard performance
measure in competitive analysis, we evaluate the performance of an online
algorithm in terms of its {\em competitiveness} (or {\em competitive
ratio}). 

An online algorithm $\A$ is \emph{$c$-competitive} on instance $\I$ if the following holds: Let $\sigma$ and $\sigma^*$ be the state computed by algorithm $\A$ and the optimal state for $\I$, respectively. Then, 
% there exists some $b \geq 0$ such that
%${\sf NSW}(\sigma) \geq \frac{\opt_\Pi(I)}{c} - b$ for every instance $I$.
${\sf NSW}(\sigma) \leq c \cdot {\sf NSW}(\sigma^*)$.
%If $b=0$ then $\mA$ is \emph{strictly} $c$-competitive.
The competitive ratio $\CR_\A(\I)$ of algorithm $\A$ on instance $\I$ is the smallest $c$ such that $\A$ is $c$-competitive on $\I$~\cite{BoE98}.
%In this \emph{Extended Abstract} we consider only the strict competitive ratio in our lower bounds.

Given a class $\mathcal{I}$ of load balancing instances, the competitive ratio $\CR_\A(\mathcal{I})$  of Algorithm $\A$ on $\mathcal{I}$ is simply given by the maximum competitive ratio of  $\A$  over all instances $\I \in \mathcal{I}$,i.e., $\CR_\A(\mathcal{I})=\sup_{\I \in \mathcal{I}}{\CR_\A(\I)}$.

\smallskip
\noindent\textbf{Greedy algorithm.} A natural algorithm proposed in \cite{AAG+95} for this problem is to assign each client to the resource yielding the minimum increase to the social welfare (ties are broken arbitrarily). This results to {\em greedy assignments}. 
Therefore, given an instance of online load balancing, an assignment of clients to resources is called a greedy assignment if the assignment of a client to a resource minimizes the total cost of the instance revealed up to the time of its appearance.
\iffalse
\paragraph{Online setting and Competitive Ratio.}
From the algorithmic point of view, load balancing has been studied
extensively, including papers studying online versions of the
problem (e.g., \cite{AAWY97,AAS01,AAG+95,AE05,CW75,CC76,PW93,SWW95,STZ04}). 
 A natural greedy algorithm proposed in \cite{AAG+95} for
this problem is to assign each client to that server that yields the
minimum increase to the total latency (ties are broken arbitrarily).
This results to {\em greedy assignments}. Given an instance of
online load balancing, an assignment of clients to servers is called
a greedy assignment if the assignment of a client to a server
minimizes the increase in the cost of the instance revealed up to
the time of its appearance. 

 The competitiveness of the greedy algorithm on an instance
is the maximum ratio of the cost of any greedy assignment over the
optimal cost and its competitiveness on a class of load balancing
instances is simply the maximum competitiveness over all instances
in the particular class.
\fi

%\paragraph{One-Round Walk.} [valutare se inserirlo]

%\input{preliminaries}

%Results:
\section{Selfish Load Balancing}\label{sec_selfishLB}
In this section we focus on selfish load balancing. 
In particular, in Subsection \ref{subsec_nash_weighted} we deal with the analysis of the price of anarchy in weighted load balancing games, in Subsection \ref{subsec_nash_unweighted} we consider the subclass of unweighted load balancing games, while in Subsection \ref{subsec_nash_nonatomic} we analyze the price of anarchy of non-atomic load balancing games. %, thus obtaining (by exploiting a result in \cite{ADKTWR08}) also bounds to the price of stability for the atomic setting.
%Finally, Subsection \ref{subsec_nash_computation} is devoted the the computational complexity of Nash equilibria.

\subsection{The {\sf NPoA} for Weighted Load Balancing Games}\label{subsec_nash_weighted}
We first provide an upper bound to the Nash price of anarchy of weighted load balancing games.
\begin{theorem}\label{thm_w_upp}
Let $\mathcal{C}$ be a class of latency functions. The Nash price of anarchy of weighted load balancing games with latency functions in $\mathcal{C}$ is 
$
{\sf NPoA}({\sf WLB}(\mathcal{C}))\leq \sup_{k_1\geq o_1>  0,o_2>k_2\geq 0,f_1,f_2\in \mathcal{C}}\left(\frac{f_1(k_1+o_1)}{f_1(o_1)}\right)^{\frac{(o_2-k_2)o_1}{k_1 o_2-k_2 o_1}}\left(\frac{f_2(k_2+o_2)}{f_2(o_2)}\right)^{\frac{(k_1-o_1)o_2}{k_1 o_2-k_2 o_1}}.
$
\end{theorem}

\begin{proof}
Let $\CG\in {\sf WLB}(\mathcal{C})$ be a weighted load balancing game with latency functions in $\mathcal{C}$, and let $\sg$ and $\sg^*$ be a worst case pure Nash equilibrium and an optimal strategy profile of $\CG$, respectively. Let $k_j$ denote $k_j(\sg)$ and $o_j$ denote $k_j(\sg^*)$. 

Since $\sg$ is a pure Nash equilibrium, we have that $cost_i(\sg)\leq cost_i(\sg_{-i},\sigma_i^*)$. Thus, we get 
$
\prod_{i\in \N}cost_i(\sg)^{w_i}\leq \prod_{i\in \N}cost_i(\sg_{-i},\sigma_i^*)^{w_i}.
$
Since $cost_i(\sg)=\ell_{\sigma_i}(k_{\sigma_i})$ and $cost_i(\sg_{-i},\sigma_i^*)\leq \ell_{\sigma_i^*}(k_{\sigma_{i}^*}+w_i)$, it holds that
$
\prod_{i\in \N}cost_i(\sg)^{w_i}=\prod_{i\in \N}\ell_{\sigma_i}(k_{\sigma_i})^{w_i}=\prod_{j\in R(\sg)}\ell_{j}(k_{j})^{\sum_{i:j=\sigma_i}w_i}=\prod_{j\in R(\sg)}\ell_{j}(k_{j})^{k_j}$
and
$\prod_{i\in \N}cost_i(\sg_{-i},\sigma_i^*)^{w_i}\leq  \prod_{i\in \N}\ell_{\sigma_i^*}(k_{\sigma_i^*}+w_i)^{w_i}\leq \prod_{i\in \N}\ell_{\sigma_i^*}(k_{\sigma_i^*}+o_{\sigma_i^*})^{w_i}=\prod_{j\in R(\sg^*)}\ell_{j}(k_{j}+o_j)^{\sum_{i:j=\sigma_i^*}w_i}=\prod_{j\in R(\sg^*)}\ell_{j}(k_{j}+o_j)^{o_j}.
$
By putting together the above inequalities we get
\begin{align}
\prod_{j\in R(\sg)}\ell_{j}(k_{j})^{k_j}&= \prod_{i\in \N}cost_i(\sg)^{w_i}\leq \prod_{i\in \N}cost_i(\sg_{-i},\sigma_i^*)^{w_i}\leq \prod_{j\in R(\sg^*)}\ell_{j}(k_{j}+o_j)^{o_j}.\label{w_form_nash} 
\end{align}
By exploiting the properties of the logarithmic function and by using (\ref{w_form_nash}), we obtain\small
\begin{align}
&\ln\left({\sf NPoA}(\CG)\right)=\ln\left(\frac{\left({\prod_{j\in R(\sg)}\ell_j(k_{j})^{k_{j}}}\right)^{\frac{1}{\sum_{i\in N}w_i}}}{\left({\prod_{j\in R(\sg^*)}\ell_j(o_{j})^{o_{j}}}\right)^{\frac{1}{\sum_{i\in N}w_i}}}\right)\nonumber\\
&\leq \ln\left(\frac{\left({\prod_{j\in R(\sg^*)}\ell_j(k_{j}+o_j)^{o_{j}}}\right)^{\frac{1}{\sum_{i\in N}w_i}}}{\left({\prod_{j\in R(\sg^*)}\ell_j(o_{j})^{o_{j}}}\right)^{\frac{1}{\sum_{i\in N}w_i}}}\right)=\frac{\sum_{j\in R(\sg^*)}o_j\left(\ln(\ell_j(k_j+o_j))-\ln(\ell_j(o_j))\right)}{\sum_{i\in N}w_i},\label{w_form_4b}
\end{align}
\normalsize
Since $\sum_{i\in \N}w_i=\sum_{j\in R}k_j=\sum_{j\in R}o_j$, we have that (\ref{w_form_4b}) is upper bounded by the optimal solution of the following optimization problem on some new linear variables $(\alpha_j)_{j\in R}$ (as (\ref{w_form_4b}) is the solution obtained by setting $\alpha=1$ for each $j\in R$):\small
\begin{align}
\max \quad & \frac{\sum_{j\in R(\sg^*)}\alpha_j o_j\left(\ln(\ell_j(k_j+o_j))-\ln(\ell_j(o_j))\right)}{\sum_{j\in R}\alpha_j k_j}\label{w_form_5b}\\
\text{s.t.} \quad &\sum_{j\in R}\alpha_j k_j=\sum_{j\in R}\alpha_j o_j,\quad  \alpha_j\geq 0\ \forall j\in R.\nonumber
\end{align}
\normalsize
\begin{fact}\label{prel_lem_w}
The maximum value of the optimization problem considered in (\ref{w_form_5b}) is at most 
$
\sup_{\substack{k_1\geq  o_1> 0,\\o_2>k_2\geq 0,\\f_1,f_2\in \mathcal{C}}}\frac{(o_2-k_2)o_1\left(\ln(f_1(k_1+o_1))-\ln(f_1(o_1))\right)+(k_1-o_1)o_2\left(\ln(f_2(k_2+o_2))- \ln(f_2(o_2))\right)}{k_1 o_2-k_2 o_1}.
$
\end{fact}
By Fact \ref{prel_lem_w}, and by continuing from (\ref{w_form_4b}), we have that the upper bound provided in Fact \ref{prel_lem_w} is higher or equal than $\ln({\sf NPoA}({\sf LB}))$. Thus, by exponentiating such inequality, we get
$
{\sf NPoA}(\CG)\leq \sup_{\substack{k_1\geq o_1> 0,\\o_2>k_2\geq 0,\\f_1,f_2\in \mathcal{C}}}\left(\frac{f_1(k_1+o_1)}{f_1(o_1)}\right)^{\frac{(o_2-k_2)o_1}{k_1 o_2-k_2 o_1}}\left(\frac{f_2(k_2+o_2)}{f_2(o_2)}\right)^{\frac{(k_1-o_1)o_2}{k_1 o_2-k_2 o_1}}.
$
Hence, by the arbitrariness of $\CG\in {\sf WLB(\mathcal{C})}$, the claim follows. \qed
\end{proof}
In the following theorem we show that the upper bound derived in Theorem \ref{thm_w_upp} is tight under mild assumptions on the latency functions.

\begin{theorem}\label{thm_w_low}
Let $\mathcal{C}$ be a class of latency functions. 

(i) If $\mathcal{C}$ is abscissa-scaling and ordinate-scaling, then 
$
{\sf NPoA}({\sf WLB}(\mathcal{C}))\geq  \sup_{k_1\geq o_1> 0,o_2>k_2\geq 0,f_1,f_2\in \mathcal{C}}\left(\frac{f_1(k_1+o_1)}{f_1(o_1)}\right)^{\frac{(o_2-k_2)o_1}{k_1 o_2-k_2 o_1}}\left(\frac{f_2(k_2+o_2)}{f_2(o_2)}\right)^{\frac{(k_1-o_1)o_2}{k_1 o_2-k_2 o_1}}.
$

(ii) If $\mathcal{C}$ is abscissa-scaling, ordinate-scaling, and unbounded-including, the previous inequality holds even for symmetric weighted load balancing games. 
\end{theorem}
\begin{proof}[Sketch of the proof]
We show part (ii) of the claim only (the proof of part (i) resorts to similar arguments and is left to the appendix). Let us assume that $\mathcal{C}$ is abscissa-scaling, ordinate-scaling, and unbounded-including. In order to prove part (ii), we equivalently show that for any $M<\sup_{k_1\geq o_1> 0,o_2>k_2\geq 0,f_1,f_2\in \mathcal{C}}\left(\frac{f_1(k_1+o_1)}{f_1(o_1)}\right)^{\frac{(o_2-k_2)o_1}{k_1 o_2-k_2 o_1}}\left(\frac{f_2(k_2+o_2)}{f_2(o_2)}\right)^{\frac{(k_1-o_1)o_2}{k_1 o_2-k_2 o_1}}$ there exists a game $\CG\in {\sf WLB}(\mathcal{C})$ such that ${\sf NPoA}(\CG)>M$. 

Let $f_1,f_2\in\mathcal{C}$, $k_1,k_2,o_1,o_2\geq 0$ such that $k_1\geq o_1>0,o_2>k_2\geq 0$, and a sufficiently small $\epsilon>0$ such that $\left(\frac{f_1(k_1+o_1)}{f_1(o_1)}\right)^{\frac{(o_2-k_2)o_1}{k_1 o_2-k_2 o_1}}\left(\frac{f_2(k_2+o_2)}{f_2(o_2)}\right)^{\frac{(k_1-o_1)o_2}{k_1 o_2-k_2 o_1}}> M+\epsilon.$ 
Let $f,g\in\mathcal{C}$ be such that $f(x):=f_1(o_1x)$ and $g(x):=f_2(o_2x)$,  and let $k:=k_1/o_1$ and $h:=k_2/o_2$. Since $\left(\frac{f_1(k_1+o_1)}{f_1(o_1)}\right)^{\frac{(o_2-k_2)o_1}{k_1 o_2-k_2 o_1}}\left(\frac{f_2(k_2+o_2)}{f_2(o_2)}\right)^{\frac{(k_1-o_1)o_2}{k_1 o_2-k_2 o_1}}=\left(\frac{f(k+1)}{f(1)}\right)^{\frac{1-h}{k-h}}\left(\frac{g(h+1)}{g(1)}\right)^{\frac{k-1}{k-h}}$ we have that 
%$\left(\frac{f(k+1)}{f(1)}\right)^{\frac{1}{k}}\geq M+\epsilon$, with $f\in\mathcal{C}$ and $k\geq 1$. %.Thus we can assume that 
\begin{equation}\label{w_low_form_0}
\left(\frac{f(k+1)}{f(1)}\right)^{\frac{1-h}{k-h}}\left(\frac{g(h+1)}{g(1)}\right)^{\frac{k-1}{k-h}}> M+\epsilon\text{, for some } f,g\in\mathcal{C},\ k\geq 1,\text{ and }h<1.
\end{equation}
Observe that $f$ and $g$ can be chosen in such a way that they are non-constant functions. Indeed, if one of them is constant, it is sufficient replacing it with an arbitrary non-constant function, so that  (\ref{w_low_form_0}) holds as well. Since $\mathcal{C}$ is unbounded-including and $f,g$ are non-constant, we have that $\lim_{x\rightarrow \infty}f(x)=\lim_{x\rightarrow \infty}g(x)=\infty$. 

We consider the case $h>0$ only (the case $h=0$ is analogue and is left to the appendix). Given two integers $m\geq 3$ and $s\geq 1$, let $\CG(m,s)$ be a symmetric weighted load balancing game where the resources are partitioned into $2m$ groups $R_1,R_2,R_3\ldots, R_{2m}$. Each group $R_j$ has $s^{j-1}$ resources and the latency function of each resource $r\in R_{j}$ is defined as $\ell_r(x):=\alpha_j \hat{f}_j\left(\beta_j x\right)$ with 
\begin{align}
&\hat{f}_j:=
\begin{cases}
f & \text{ if }j\leq m-1\\
g & \text{ if }j\geq m
\end{cases},\quad 
\beta_j:=
\begin{cases}
\left(\frac{s}{k}\right)^{j-1} & \text{ if }j\leq m-1\\
\left(\frac{s}{h}\right)^{j-m}\left(\frac{s}{k}\right)^{m-1} & \text{ if }m\leq j\leq 2m
\end{cases},\\
&\alpha_j:=
\begin{cases}
\left(\frac{f(k)}{f(k+1)}\right)^{j-1} & \text{ if }j\leq m-1\\
\left(\frac{g(h)}{g(h+1)}\right)^{j-m}\left(\frac{f(k)}{g(h+1)}\right)\left(\frac{f(k)}{f(k+1)}\right)^{m-2} & \text{ if }m\leq j\leq 2m-1\\
\frac{g(h)}{g(1)}\left(\frac{g(h)}{g(h+1)}\right)^{m-1}\left(\frac{f(k)}{g(h+1)}\right)\left(\frac{f(k)}{f(k+1)}\right)^{m-2} & \text{ if }j=2m
\end{cases}.
\end{align}
The set of players $\N$ is partitioned into $2m-1$ sets $\N_1,\N_2,\ldots, \N_{2m-1}$, and each group $\N_j$ has $s^{j}$ players having weight $w_j:=1/\beta_{j+1}$. Let $\sg$ be the strategy profile in which, for any $j\in [2m-1]$, each resource of group $R_{j}$ is selected by exactly $s$ players of group $\N_j$ (see Figure \ref{fig:boat2}.a). One can show that, for any integer $m\geq 3$, there exists a sufficiently large $s_m$ such that $\sg$ is a pure Nash equilibrium of the game $\CG(m,s_m)$  (see the appendix for a complete proof). 

\begin{figure}
 \includegraphics[width=\linewidth]{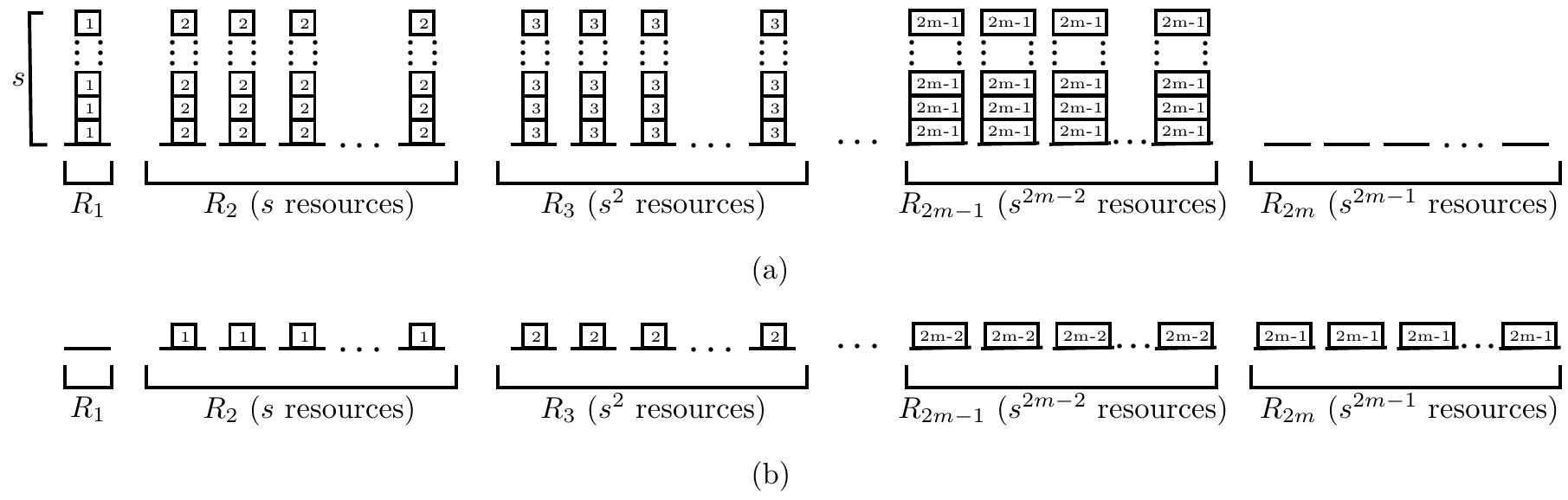}
 \caption{The $\CG$ used in the proof of Theorem \ref{thm_w_low}. Columns represent resources and squares represent players (number $j$ inside a square means that the player belongs to group $N_j$). (a): a Nash equilibrium $\bm\sigma$; (b): the strategy profile ${\bm\sigma}^*$.}\label{fig:boat2}
\end{figure}

Now, let $\sg^*$ be the strategy profile of $\CG(m,s_m)$ in which, for any $j\in [2m-1]$, each resource of group $R_{j+1}$ is selected by exactly one player of group $\N_j$ (see Figure \ref{fig:boat2}.b). By exploiting the definitions of $\alpha_j$,$\beta_j$, $\hat{f}_j$, $w_j$, and $N_j$, and by choosing a sufficiently large $m$, one can show that the following inequalities hold (see the appendix for a complete proof):
$
\frac{{\sf NSW}(\sg)}{{\sf NSW}(\sg^*)}\geq \lim_{m\rightarrow \infty}\left(\frac{\prod_{j=1}^{2m-1}\left(\alpha_j\hat{f}_j\left(\beta_j s_m w_j\right)\right)^{|\N_j|w_j}}{\prod_{j=2}^{2m}\left(\alpha_{j}\hat{f}_j\left(\beta_j w_{j-1}\right)\right)^{|\N_{j-1}|w_{j-1}}}\right)^{\frac{1}{\sum_{j=1}^{2m-1}|\N_j|w_j}}-\epsilon=\left(\frac{f(k+1)}{f(1)}\right)^{\frac{1-h}{k-h}}\left(\frac{g(h+1)}{g(1)}\right)^{\frac{k-1}{k-h}}-\epsilon>M+\epsilon-\epsilon=M,
$
thus showing part (ii) of the claim. \qed
\end{proof}

When considering functions belonging to the class $\mathcal{P}(p)$ of polynomials of maximum degree $p$, the following technical lemma holds.
\begin{lemma}\label{lemma_PoA_weighted_Polynomial latency functions}
$\sup_{\substack{k_1\geq o_1> 0,\\o_2>k_2\geq 0,\\f_1,f_2\in \mathcal{P}(p)}}\left(\frac{f_1(k_1+o_1)}{f_1(o_1)}\right)^{\frac{(o_2-k_2)o_1}{k_1 o_2-k_2 o_1}}\left(\frac{f_2(k_2+o_2)}{f_2(o_2)}\right)^{\frac{(k_1-o_1)o_2}{k_1 o_2-k_2 o_1}}= 2^p.$
%The Nash price of anarchy of weighted load balancing games with polynomial latency functions (even for symmetric games) of maximum degree $p$ is ${\sf NPoA}({\sf WLB}(\mathcal{P}(p)))={\sf NPoA}({\sf SWLB}(\mathcal{P}(p)))=2^p$. 
\end{lemma}
%(observe that polynomial latency functions are abscissa-scaling, ordinate-scaling, and all-unbounded)
Given Lemma \ref{lemma_PoA_weighted_Polynomial latency functions}, and since the class of polynomial latency functions is ordinate-scaling, abscissa-scaling, and unbounded-including, the following corollary of Theorems \ref{thm_w_upp} and \ref{thm_w_low} establishes the exact Nash price of anarchy for polynomial latency functions.
\begin{corollary}\label{cor1}
The Nash price of anarchy of weighted load balancing games with polynomial latency functions (even for symmetric games) of maximum degree $p$ is ${\sf NPoA}({\sf WLB}(\mathcal{C})) = 2^p$.
\end{corollary}
When considering identical resources with polynomial latency functions, the price of anarchy does not decrease, as shown in the following corollary of Theorem \ref{thm_w_low}.
\begin{corollary}\label{cor1b}
The Nash price of anarchy of weighted load balancing games with polynomial latency functions of maximum degree $p$ and identical resources is at least $2^p$.
\end{corollary}

\subsection{The {\sf NPoA} for Unweighted Load Balancing Games}\label{subsec_nash_unweighted}
We first provide an upper bound to the Nash price of anarchy of unweighted load balancing games.

\begin{theorem}\label{thm_unw_upp}
Let $\mathcal{C}$ be a class of latency functions. The Nash price of anarchy of unweighted load balancing games with latency functions in $\mathcal{C}$ is
$
{\sf NPoA}({\sf ULB}(\mathcal{C}))\leq \sup_{f\in\mathcal{C}, k\in \ZP, o\in  [k]}\left(\frac{f(k+1)}{f(o)}\right)^{\frac{o}{k}}.
$
\end{theorem}

We  show that the upper bound derived in Theorem \ref{thm_unw_upp} is tight if the considered latency functions are ordinate-scaling (the proof is deferred to the appendix). The following result for polynomial latency functions holds.
\begin{corollary}\label{cor2}
The Nash price of anarchy of unweighted load balancing games with polynomial latency functions of maximum degree $p$ is ${\sf NPoA}({\sf ULB}(\mathcal{C})) = 2^p$.
\end{corollary}

%By the previous corollary, we have that the price of anarchy of unweighted congestion games is as high as that of weighted ones. 

\subsection{The {\sf NPoA} for Non-Atomic Load Balancing Games}\label{subsec_nash_nonatomic}
%{\color{magenta}PoA non-atomic = PoS LBG unweighted (tight UB e LB)}

We first provide an upper bound to the Nash price of anarchy of non-atomic load balancing games.

\begin{theorem}\label{thm_non_upp}
Let $\mathcal{C}$ be a class of latency functions. The Nash price of anarchy of non-atomic load balancing games with latency functions in $\mathcal{C}$ is 
$
{\sf NPoA}({\sf NLB}(\mathcal{C}))\leq \sup_{f\in\mathcal{C},  k\geq o>0}\left(\frac{f(k)}{f(o)}\right)^{\frac{o}{k}}.
$
\end{theorem}
%%%%%%%%%%%%%%%%%%%%%%%%%%%%%%%%%%%%%%%%%%%%%%%%%%%%%%%%%
%%%%%%%%%%%%%%%%%%%%%%%%%%%%%%%%%%%%%%%%%%%%%%%%%%%%%%%%%

%\textcolor{magenta}{**TODO: la dimostrazione di questo upper bound ricalca la dimostrazione dell'upper bound per giochi non pesati, ma richiede una notazione diversa, tipica dei giochi non atomici (che non ho scritto).}

We show that the upper bound derived in Theorem \ref{thm_non_upp} is tight the considered latency functions are all-constant-including (the proof is deferred to the appendix). The following result for polynomial latency functions holds. 
\begin{corollary}\label{cor3}
The Nash price of anarchy of non-atomic load balancing games with polynomial latency functions of maximum degree $p$ (even for symmetric games) is ${\sf NPoA}({\sf NLB}(\mathcal{P}(p)))={\sf NPoA}({\sf SNLB}(\mathcal{P}(p)))=\left(e^\frac{1}{e}\right)^p\simeq (1.44)^p$. 
\end{corollary}

\section{Online load balancing}\label{sec_onlineLB}
We first provide an upper bound on the competitive ratio of the greedy algorithm.
\begin{theorem}\label{thm_onl_upp}
Let $\mathcal{C}$ be a class of quasi-log-convex functions. The competitive ratio of the greedy algorithm ${\sf G}$ applied to load balancing instances with latency functions in $\mathcal{C}$ is
$
{\sf CR}_{\sf G}({\sf WLB}(\mathcal{C}))\leq \sup_{k_1\geq  o_1> 0,o_2>k_2\geq 0,f_1,f_2\in \mathcal{C}}\left(\frac{f_1(k_1+o_1)^{k_1+o_1}}{f_1(k_1)^{k_1}f_1(o_1)^{o_1}}\right)^{\frac{o_2-k_2}{o_2 k_1-o_1k_2}}\left(\frac{f_2(k_2+o_2)^{k_2+o_2}}{f_2(k_2)^{k_2}f_2(o_2)^{o_2}}\right)^{\frac{k_1-o_1}{o_2 k_1-o_1k_2}},
$
where we set $f_2(0)^0:=1$. 
\end{theorem}
We show that, when considering the greedy algorithm, the upper bound derived in Theorem \ref{thm_onl_upp} is tight if the considered latency functions are abscissa-scaling and ordinate-scaling (the proof is deferred to the appendix). The following result for polynomial latency functions holds (the proof is deferred to the appendix). 
\begin{corollary}\label{pol_onl_cor}
The competitive ratio of the greedy algorithm applied to weighted load balancing instances with polynomial latency functions of maximum degree $p$ is ${\sf CR}_{\sf G}({\sf WLB}(\mathcal{C}))=4^p$.
\end{corollary}

We show that, when considering polynomial latency functions, the upper bound of Corollary  \ref{pol_onl_cor} is tight for any  online algorithm, i.e., we are able to provide a matching lower bound to the online load balancing problem (the proof is deferred to the appendix). 
\section{Concluding Remarks and Open Problems}\label{sec_conlusion}
To the best of our knowledge, this is the first work that adopts the NSW as a benchmarking quality measure in load balancing problems. Several open problems deserve further investigation.

First of all, our paper mostly focuses on evaluating the performance of selfish and online load balancing. 
Concerning complexity issues, it is worth noticing that, on the one hand, when considering unweighted players, an optimal configuration with respect to the NSW can be trivially computed in polynomial time by exploiting the same techniques developed in \cite{CMNV05,MS12} for the utilitarian social welfare (\cite{CMNV05,MS12} use, in turn, an approach similar to the one adopted in \cite{FPT04} for the computation of a Nash equilibrium); on the other hand, when considering weighted players, a simple reduction from the NP-complete problem $\mathsf{PARTITION}$ shows that the problem becomes NP-hard.  Therefore, an interesting open problem is that of providing polynomial time approximation algorithms for the weighted case (we notice that Corollary \ref{pol_onl_cor} provides a $4^p$-approximation algorithm for weighted load balancing instances with polynomial latency functions of maximum degree $p$).

Moreover, a natural extension of our results consists in considering other families of congestion games, being more general than the one of load balancing games, such as the family of matroid congestion games \Mycite{HRV08,dKU16}. 

Finally, it would be interesting to apply the NSW measure to other classes of games, whose performances, in the literature, have only been analysed with respect to the utilitarian and/or egalitarian social welfare functions.
%
% ---- Bibliography ----
%
% BibTeX users should specify bibliography style 'splncs04'.
% References will then be sorted and formatted in the correct style.
%
% \bibliographystyle{splncs04}
% \bibliography{mybibliography}
%
\bibliographystyle{splncs04}
\bibliography{bibliofinal}

\newpage

\appendix
\section{Missing Proofs of Subsection 3.1}
\subsection{Proof of Fact \ref{prel_lem_w}}
First of all, by exploiting the structure of the optimization problem, we can introduce the  normalization constraint $\sum_{j\in R}\alpha_j k_j=\sum_{j\in R}\alpha_j o_j=1$ without affecting the optimal value of the problem. By introducing such normalization constraint, the optimization problem becomes the following linear program:
\begin{align}
\max \quad & \sum_{j\in R(\sg^*)}\alpha_j o_j\left(\ln(\ell_j(k_j+o_j))-\ln(\ell_j(o_j))\right)\label{w_form_5c}\\
\text{s.t.} \quad &\sum_{j\in R}\alpha_j k_j=1,\quad   \sum_{j\in R}\alpha_j o_j=1,\quad \alpha_j\geq 0\ \forall j\in R.\nonumber
\end{align}
By standard arguments of linear programming, we have that an optimal solution of (\ref{w_form_5c}) is given by a vertex of the polyhedral region defined by the linear constraints of (\ref{w_form_5c}), and such vertex can be obtained by nullifying at least $|R|-2$ variables. Thus, we can assume w.l.o.g. that in an optimal solution there are at most
two variables, say $\alpha_1$ and $\alpha_2$, such that $\alpha_1\geq 0$ and $\alpha_2\geq 0$.
If both variables $\alpha_1$ and $\alpha_2$ are positive, we have that they are univocally determined by the constraints $\alpha_1k_1+\alpha_2k_2=1$ and $\alpha_1o_1+\alpha_2o_2=1$, so that 
\begin{align}
\alpha_1=\frac{o_2-k_2}{k_1 o_2-k_2 o_1}> 0,\quad \alpha_2=\frac{k_1-o_1}{k_1 o_2-k_2 o_1}> 0,\quad \alpha_j=0\ \forall j\geq 3.\label{w_form_6b}
\end{align}
By symmetry, we can assume w.l.o.g. that $k_1 o_2-k_2 o_1>0$, so that $k_1> o_1\geq 0$ and $o_2>k_2\geq 0$. 

Now, assume that one variable among $\alpha_1$ and $\alpha_2$ is null, and assume w.l.o.g. that $\alpha_2=0$. In this case, we necessarily get $k_1=o_1>0$ and $\alpha_1=1/o_1$, and the value of the objective function becomes $\ln(f_1(2o_1))-\ln(f_1(o_1))$. Anyway, we obtain the same value of the objective function by using in (\ref{w_form_5c}) the values of $\alpha_1$ and $\alpha_2$ considered in (\ref{w_form_6b}), and by setting $k_1=o_1>0$ and $o_2>k_2\geq 0$. We also observe that, if $o_1=0$ and $\alpha_1,\alpha_2>0$, the value of the objective function is $\ln(f_2(k_2+o_2))-\ln(f_2(o_2))\leq \ln(f_2(2o_2))-\ln(f_2(o_2))$, i.e., at most equal to the value of the objective function in which one of the two variables among $\alpha_1$ and $\alpha_2$ is null. Thus, we may omit the case $o_1=0$. 

We conclude that, by considering the objective function of (\ref{w_form_5c}) with the values $\alpha_1$ and $\alpha_2$ defined in (\ref{w_form_6b}), and by considering  the supremum of the objective function over $k_1\geq o_1>0$ and $o_2>k_2\geq 0$, we obtain the upper bound of the claim. 
\subsection{Proof of Theorem \ref{thm_w_low}}
First of all, we deal with part (ii) of the claim: Let us assume that $\mathcal{C}$ is abscissa-scaling, ordinate-scaling, and unbounded-including. In order to prove part (ii), we equivalently show that for any $M<\sup_{k_1\geq o_1> 0,o_2>k_2\geq 0,f_1,f_2\in \mathcal{C}}\left(\frac{f_1(k_1+o_1)}{f_1(o_1)}\right)^{\frac{(o_2-k_2)o_1}{k_1 o_2-k_2 o_1}}\left(\frac{f_2(k_2+o_2)}{f_2(o_2)}\right)^{\frac{(k_1-o_1)o_2}{k_1 o_2-k_2 o_1}}$ there exists a game $\CG\in {\sf WLB}(\mathcal{C})$ such that ${\sf NPoA}(\CG)>M$. 

Let $f_1,f_2\in\mathcal{C}$, $k_1,k_2,o_1,o_2\geq 0$ such that $k_1\geq o_1>0,o_2>k_2\geq 0$, and a sufficiently small $\epsilon>0$ such that $\left(\frac{f_1(k_1+o_1)}{f_1(o_1)}\right)^{\frac{(o_2-k_2)o_1}{k_1 o_2-k_2 o_1}}\left(\frac{f_2(k_2+o_2)}{f_2(o_2)}\right)^{\frac{(k_1-o_1)o_2}{k_1 o_2-k_2 o_1}}> M+\epsilon.$ 
Let $f,g\in\mathcal{C}$ be such that $f(x):=f_1(o_1x)$ and $g(x):=f_2(o_2x)$,  and let $k:=k_1/o_1$ and $h:=k_2/o_2$. Since $$\left(\frac{f_1(k_1+o_1)}{f_1(o_1)}\right)^{\frac{(o_2-k_2)o_1}{k_1 o_2-k_2 o_1}}\left(\frac{f_2(k_2+o_2)}{f_2(o_2)}\right)^{\frac{(k_1-o_1)o_2}{k_1 o_2-k_2 o_1}}=\left(\frac{f(k+1)}{f(1)}\right)^{\frac{1-h}{k-h}}\left(\frac{g(h+1)}{g(1)}\right)^{\frac{k-1}{k-h}}$$

we have that 
%$\left(\frac{f(k+1)}{f(1)}\right)^{\frac{1}{k}}\geq M+\epsilon$, with $f\in\mathcal{C}$ and $k\geq 1$. %.Thus we can assume that 
\begin{equation}\label{w_low_form_0_app}
\left(\frac{f(k+1)}{f(1)}\right)^{\frac{1-h}{k-h}}\left(\frac{g(h+1)}{g(1)}\right)^{\frac{k-1}{k-h}}> M+\epsilon\text{, for some } f,g\in\mathcal{C},\ k\geq 1,\text{ and }h<1.
\end{equation}
Observe that $f$ and $g$ can be chosen in such a way that they are non-constant functions. Indeed, if one of them is constant, it is sufficient replacing it with an arbitrary non-constant function, so that  (\ref{w_low_form_0_app}) holds as well. Since $\mathcal{C}$ is unbounded-including and $f,g$ are non-constant, we have that $\lim_{x\rightarrow \infty}f(x)=\lim_{x\rightarrow \infty}g(x)=\infty$. 

First of all, we assume that $h>0$. Given two integers $m\geq 3$ and $s\geq 1$, let $\CG(m,s)$ be a symmetric weighted load balancing game where the resources are partitioned into $2m$ groups $R_1,R_2,R_3\ldots, R_{2m}$. Each group $R_j$ has $s^{j-1}$ resources and the latency function of each resource $r\in R_{j}$ is defined as $\ell_r(x):=\alpha_j \hat{f}_j\left(\beta_j x\right)$ with 
\begin{align}
&\hat{f}_j:=
\begin{cases}
f & \text{ if }j\leq m-1\\
g & \text{ if }j\geq m
\end{cases},\quad 
\beta_j:=
\begin{cases}
\left(\frac{s}{k}\right)^{j-1} & \text{ if }j\leq m-1\\
\left(\frac{s}{h}\right)^{j-m}\left(\frac{s}{k}\right)^{m-1} & \text{ if }m\leq j\leq 2m
\end{cases},\\
&\alpha_j:=
\begin{cases}
\left(\frac{f(k)}{f(k+1)}\right)^{j-1} & \text{ if }j\leq m-1\\
\left(\frac{g(h)}{g(h+1)}\right)^{j-m}\left(\frac{f(k)}{g(h+1)}\right)\left(\frac{f(k)}{f(k+1)}\right)^{m-2} & \text{ if }m\leq j\leq 2m-1\\
\frac{g(h)}{g(1)}\left(\frac{g(h)}{g(h+1)}\right)^{m-1}\left(\frac{f(k)}{g(h+1)}\right)\left(\frac{f(k)}{f(k+1)}\right)^{m-2} & \text{ if }j=2m
\end{cases}.
\end{align}
\normalsize
The set of players $\N$ is partitioned into $2m-1$ sets $\N_1,\N_2,\ldots, \N_{2m-1}$, and each group $\N_j$ has $s^{j}$ players having weight $w_j:=1/\beta_{j+1}$. Let $\sg$ be the strategy profile in which, for any $j\in [2m-1]$, each resource of group $R_{j}$ is selected by exactly $s$ players of group $\N_j$ (see Figure \ref{fig:boat2}.a). Observe that, by construction of $\alpha_j,\beta_j,w_j$, the following properties hold:\small
\begin{equation}\label{prop_low_w_app}
\begin{cases}
\alpha_{j}f(k)=\alpha_{j+1}f(k+1) &\text{if }j\leq m-2\\
\alpha_{j}f(k)=\alpha_{j+1}g(h+1) &\text {if }j=m-1\\
\alpha_{j}g(h)=\alpha_{j+1}g(h+1) &\text {if }m\leq j\leq 2m-2\\
\alpha_{j}g(h)=\alpha_{j+1}g(1)& \text {if }j=2m-1
\end{cases},
\begin{cases}
\beta_{j} w_j s=k,\ w_j|N_j|=k^j&\text{if }j\leq m-1\\
\beta_{j} w_j s=h,\ w_j|N_j|=h^{j+1-m} k^{m-1} &\text {if }m\leq j\leq 2m-1\\
\beta_{j+1} w_j=1&\text{if }j\leq 2m-1
\end{cases}
\end{equation}\normalsize
We now show that, by choosing a sufficiently large $s$, the strategy profile $\sg$ is a pure Nash equilibrium of $\CG(m,s)$. 
%\begin{lemma}\label{lem_weig}
%By taking a sufficiently large $s$, the strategy profile $\sg$ is a pure Nash equilibrium . 
%\end{lemma}
%\begin{proof}
Let $j\in [2m-1]$, $t\in [2m]$, and $i$ be an arbitrary player selecting a resource $r_j$ of group $R_j$ in the strategy profile $\sg$, and assume that she deviates to a resource $r_t$ of group $R_{t}$. We have three cases:
\begin{description}
\item[$\bullet\ {\bf t\bm=j\bm+1}$:] First of all, assume that $j\leq m-2$. By using (\ref{prop_low_w_app}), we get $
cost_i(\sg)=\ell_{r_j}(k_{r_j}(\sg))=\alpha_j\hat{f}_j\left(\beta_j s w_j\right)=\alpha_jf\left(k\right)=\alpha_{j+1}f\left(k+1\right)=\alpha_{j+1}f\left(\beta_{j+1} s w_{j+1}+\beta_{j+1}w_j\right)=\alpha_{j+1}\hat{f}_{j+1}\left(\beta_{j+1}(s w_{j+1}+w_j)\right)=\ell_{r_h}(k_{r_h}(\sg_{-i},\{r_t\}))=cost_i(\sg_{-i},\{r_t\})$. The cases $j=m-1$, $m\leq j\leq 2m-2$, and $j=2m-1$ can be separately considered by exploiting (\ref{prop_low_w_app}), so that one can analogously show  $cost_i(\sg)=\alpha_j\hat{f}_j(\beta_j s w_j)=\alpha_{j+1}\hat{f}_{j+1}(\beta_{j+1} (s w_{j+1}+w_j))=cost_i(\sg_{-i},\{r_t\})$, where we set $w_{2m}:=0$. 
\item[$\bullet\ \bf t\bm\leq j:$] From the previous case, we have that if one player is playing a resource at some level $l$, and deviates to some resource at level $l+1$, her cost does not change. Thus, we necessarily have that the cost of each resource in strategy profile $\sg$ is a non-increasing function of the level $l\in [2m]$ which it belongs to. Thus, since $t\leq j$, we necessarily have that  $cost_i(\sg)\leq cost_i(\sg_{-i},\{r_t\})$. 

\item[$\bullet\ \bf t\bm>j\bm+1:$]If we consider the asymptotic behaviour of $cost_i(\sg)$ and $cost_i(\sg_{-i},\{r_t\})$ with respect to parameter $s$, we get $cost_i(\sg)=\alpha_j\hat{f}_j(\beta_j sw_j)=\alpha_j\hat{f}_j(\Theta(s^{j-1}\cdot s\cdot  s^{-j}))=\Theta(1)$, thus $cost_i(\sg)$ does not depend on $s$; furthermore, we get $cost_i(\sg_{-i},\{r_t\})\geq \alpha_j\hat{f}_j(\beta_t w_{j+1})=\alpha_j\hat{f}_j(\Theta(s^{t-1}s^{-j}))\geq \alpha_j\hat{f}_j(\Theta(s))$, thus, since $\lim_{x\rightarrow \infty}\hat{f}(x)=\infty$, we have that $cost_i(\sg_{-i},\{r_t\})$ can be arbitrarily large as $s$ increases. We conclude that, by taking a sufficiently large $s$, we get $cost_i(\sg)\leq cost_i(\sg_{-i},\{r_t\})$ for any $j$ and $t>j+1$. 
\end{description}
The previous case-analysis shows that player $i$ does not improve her cost after deviating in favour of any resource $r_t$ at level $t$, for any $t\in [2m]$, and thus that $\sg$ is a pure Nash equilibrium of $\CG(m,s)$. 
%\end{proof}
For any integer $m\geq 3$, let $s_m$ be a sufficiently large integer such that (according to the previous case-analysis) $\sg$ is a pure Nash equilibrium of the game $\CG(m,s_m)$. 

Now, let $\sg^*$ be the strategy profile of $\CG(m,s_m)$ in which, for any $j\in [2m-1]$, each resource of group $R_{j+1}$ is selected by exactly one player of group $\N_j$ (see Figure \ref{fig:boat2}.b). By exploiting the definitions of $\alpha_j$,$\beta_j$, $\hat{f}_j$, $w_j$, and $N_j$, we have that:\small
\begin{align}
&{\sf NPoA}(\CG(m,s_m))\nonumber\\
\geq &\frac{{\sf NSW}(\sg)}{{\sf NSW}(\sg^*)}\nonumber\\
=&\left(\frac{\prod_{j=1}^{2m-1}\left(\alpha_j\hat{f}_j\left(\beta_j s_m w_j\right)\right)^{|\N_j|w_j}}{\prod_{j=2}^{2m}\left(\alpha_{j}\hat{f}_j\left(\beta_j w_{j-1}\right)\right)^{|\N_{j-1}|w_{j-1}}}\right)^{\frac{1}{\sum_{j=1}^{2m-1}|\N_j|w_j}}\nonumber\\
=&\left(\frac{\left(\prod_{j=1}^{m-1}\left(\alpha_jf\left(k\right)\right)^{|\N_j|w_j} \right)\left(\prod_{j=m}^{2m-1}\left(\alpha_j g\left(h\right)\right)^{|\N_j|w_j}\right)}{\left(\prod_{j=2}^{m-1}\left(\alpha_{j}f\left(1\right)\right)^{|\N_{j-1}|w_{j-1}}\right)\left(\prod_{j=m}^{2m}\left(\alpha_j g\left(1\right)\right)^{|\N_{j-1}|w_{j-1}}\right)}\right)^{\frac{1}{\sum_{j=1}^{2m-1}|\N_j|w_j}}\label{w_low_form_3-_app}\\
=&\left(\frac{\left(\prod_{j=1}^{m-1}\left(\alpha_jf\left(k\right)\right)^{k^j}\right) \left(\prod_{j=m}^{2m-1}\left(\alpha_j g\left(h\right)\right)^{h^{j+1-m}k^{m-1}}\right)}{\left(\prod_{j=2}^{m-1}\left(\alpha_{j}f\left(1\right)\right)^{k^{j-1}}\right)\left(\prod_{j=m}^{2m}\left(\alpha_j g\left(1\right)\right)^{h^{j-m}k^{m-1}}\right)}\right)^{\frac{1}{\sum_{j=1}^{2m-1}|\N_j|w_j}}\nonumber\\
=&\left(\frac{\left(\prod_{j=1}^{m-2}\left(\alpha_{j+1}f\left(k+1\right)\right)^{k^j}\right) \left(\prod_{j=m-1}^{2m-2}\left(\alpha_{j+1} g\left(h+1\right)\right)^{h^{j+1-m}k^{m-1}}\right)\left(\alpha_{2m} g\left(1\right)\right)^{h^{m}k^{m-1}}}{\left(\prod_{j=2}^{m-1}\left(\alpha_{j}f\left(1\right)\right)^{k^{j-1}}\right)\left(\prod_{j=m}^{2m}\left(\alpha_j g\left(1\right)\right)^{h^{j-m}k^{m-1}}\right)}\right)^{\frac{1}{\sum_{j=1}^{2m-1}|\N_j|w_j}}\label{w_low_form_3_app}\\
=&\left(\frac{\left(\prod_{j=1}^{m-2}\left(\alpha_{j+1}f\left(k+1\right)\right)^{k^j}\right) \left(\prod_{j=m-1}^{2m-2}\left(\alpha_{j+1} g\left(h+1\right)\right)^{h^{j+1-m}k^{m-1}}\right)\left(\alpha_{2m} g\left(1\right)\right)^{h^{m}k^{m-1}}}{\left(\prod_{j=1}^{m-2}\left(\alpha_{j+1}f\left(1\right)\right)^{k^{j}}\right)\left(\prod_{j=m-1}^{2m-1}\left(\alpha_{j+1} g\left(1\right)\right)^{h^{j+1-m}k^{m-1}}\right)}\right)^{\frac{1}{\sum_{j=1}^{2m-1}|\N_j|w_j}}\nonumber\\
=&\left(\frac{\left(\prod_{j=1}^{m-2}\left(\alpha_{j+1}f\left(k+1\right)\right)^{k^j}\right) \left(\prod_{j=m-1}^{2m-2}\left(\alpha_{j+1} g\left(h+1\right)\right)^{h^{j+1-m}k^{m-1}}\right)}{\left(\prod_{j=1}^{m-2}\left(\alpha_{j+1}f\left(1\right)\right)^{k^{j}}\right)\left(\prod_{j=m-1}^{2m-2}\left(\alpha_{j+1} g\left(1\right)\right)^{h^{j+1-m}k^{m-1}}\right)}\right)^{\frac{1}{\sum_{j=1}^{m-2}k^j+\sum_{j=m-1}^{2m-1}h^{j+1-m}k^{m-1}}}\nonumber\\
=&\left(\left(\prod_{j=1}^{m-2}\left(\frac{f(k+1)}{f(1)}\right)^{k^j}\right)\left(\prod_{j=m-1}^{2m-2}\left(\frac{g(h+1)}{g(1)}\right)^{h^{j+1-m}k^{m-1}}\right)\right)^{\frac{1}{\sum_{j=1}^{m-2}k^j+\sum_{j=m-1}^{2m-1}h^{j+1-m}k^{m-1}}}\nonumber\\
=&\left(\left(\frac{f(k+1)}{f(1)}\right)^{\sum_{j=1}^{m-2}k^j}\left(\frac{g(h+1)}{g(1)}\right)^{\sum_{j=m-1}^{2m-2}h^{j+1-m}k^{m-1}}\right)^{\frac{1}{\sum_{j=1}^{m-2}k^j+\sum_{j=m-1}^{2m-1}h^{j+1-m}k^{m-1}}},\label{w_low_form_4_app}
\end{align}
\normalsize
where (\ref{w_low_form_3-_app}) and (\ref{w_low_form_3_app}) come from (\ref{prop_low_w_app}). We have two cases: $k>1$ and $k=1$. If $k>1$, by continuing from (\ref{w_low_form_4_app}) and by considering a sufficiently large $m$, we get
\begin{align}
&\left(\left(\frac{f(k+1)}{f(1)}\right)^{\sum_{j=1}^{m-2}k^j}\left(\frac{g(h+1)}{g(1)}\right)^{\sum_{j=m-1}^{2m-2}h^{j+1-m}k^{m-1}}\right)^{\frac{1}{\sum_{j=1}^{m-2}k^j+\sum_{j=m-1}^{2m-1}h^{j+1-m}k^{m-1}}}\nonumber\\
=&\left(\left(\frac{f(k+1)}{f(1)}\right)^{\frac{k^{m-1}-k}{k-1}}\left(\frac{g(h+1)}{g(1)}\right)^{k^{m-1}\left(\frac{1-h^{m}}{1-h}\right)}\right)^{\frac{1}{\frac{k^{m-1}-k}{k-1}+k^{m-1}\left(\frac{1-h^{m+1}}{1-h}\right)}}\nonumber\\
=&\left(\frac{f(k+1)}{f(1)}\right)^{\frac{\frac{k^{m-1}-k}{k-1}}{\frac{k^{m-1}-k}{k-1}+k^{m-1}\left(\frac{1-h^{m+1}}{1-h}\right)}}\left(\frac{g(h+1)}{g(1)}\right)^{\frac{k^{m-1}\left(\frac{1-h^{m}}{1-h}\right)}{\frac{k^{m-1}-k}{k-1}+k^{m-1}\left(\frac{1-h^{m+1}}{1-h}\right)}}\nonumber\\
=&\left(\frac{f(k+1)}{f(1)}\right)^{\frac{\frac{1-h}{1-h^{m+1}}}{\frac{1-h}{1-h^{m+1}}+k^{m-1}\left(\frac{k-1}{k^{m-1}-k}\right)}}\left(\frac{g(h+1)}{g(1)}\right)^{\frac{k^{m-1}\left(\frac{k-1}{k^{m-1}-k}\right)\left(\frac{1-h^m}{1-h^{m+1}}\right)}{\frac{1-h}{1-h^{m+1}}+k^{m-1}\left(\frac{k-1}{k^{m-1}-k}\right)}}\nonumber\\
\geq &\lim_{m\rightarrow \infty} \left(\frac{f(k+1)}{f(1)}\right)^{\frac{\frac{1-h}{1-h^{m+1}}}{\frac{1-h}{1-h^{m+1}}+k^{m-1}\left(\frac{k-1}{k^{m-1}-k}\right)}}\left(\frac{g(h+1)}{g(1)}\right)^{\frac{k^{m-1}\left(\frac{k-1}{k^{m-1}-k}\right)\left(\frac{1-h^m}{1-h^{m+1}}\right)}{\frac{1-h}{1-h^{m+1}}+k^{m-1}\left(\frac{k-1}{k^{m-1}-k}\right)}}-\epsilon\label{w_low_form_5-_app}\\
=&\left(\frac{f(k+1)}{f(1)}\right)^{\frac{1-h}{(1-h)+(k-1)}}\left(\frac{g(h+1)}{g(1)}\right)^{\frac{k-1}{(1-h)+(k-1)}}-\epsilon\label{w_low_form_5_app}\\
=&\left(\frac{f(k+1)}{f(1)}\right)^{\frac{1-h}{k-h}}\left(\frac{g(h+1)}{g(1)}\right)^{\frac{k-1}{k-h}}-\epsilon\nonumber\\
> & M+\epsilon-\epsilon\label{w_low_form_6_app}\\
= &M,\label{w_low_form_7_app}
\end{align}
where (\ref{w_low_form_5-_app}) holds if $m$ is sufficiently large, (\ref{w_low_form_5_app}) comes from the fact that $k>1$ and $h<1$, and (\ref{w_low_form_6_app}) comes from (\ref{w_low_form_0_app}). 

If $k=1$, by continuing from (\ref{w_low_form_4_app}), we get:
\begin{align}
&\left(\left(\frac{f(k+1)}{f(1)}\right)^{\sum_{j=1}^{m-2}k^j}\left(\frac{g(h+1)}{g(1)}\right)^{\sum_{j=m-1}^{2m-2}h^{j+1-m}k^{m-1}}\right)^{\frac{1}{\sum_{j=1}^{m-2}k^j+\sum_{j=m-1}^{2m-1}h^{j+1-m}k^{m-1}}}\nonumber\\
=&\left(\frac{f(k+1)}{f(1)}\right)^{\frac{m-2}{m-2+\frac{1-h^{m+1}}{1-h}}}\left(\frac{g(h+1)}{g(1)}\right)^{\frac{\frac{1-h^{m}}{1-h}}{m-2+\frac{1-h^{m+1}}{1-h}}}\nonumber\\
\geq &\lim_{m\rightarrow \infty}\left(\frac{f(k+1)}{f(1)}\right)^{\frac{m-2}{m-2+\frac{1-h^{m+1}}{1-h}}}\left(\frac{g(h+1)}{g(1)}\right)^{\frac{\frac{1-h^{m}}{1-h}}{m-2+\frac{1-h^{m+1}}{1-h}}}-\epsilon\label{w_low_form_8-_app}\\
=&\left(\frac{f(k+1)}{f(1)}\right)^1\left(\frac{g(h+1)}{g(1)}\right)^0-\epsilon\nonumber\\
=&\left(\frac{f(k+1)}{f(1)}\right)^{\frac{1-h}{k-h}}\left(\frac{g(h+1)}{g(1)}\right)^{\frac{k-1}{k-h}}-\epsilon\label{w_low_form_8_app}\\
> & M+\epsilon-\epsilon\label{w_low_form_9_app}\\
= &M,\label{w_low_form_10_app}
\end{align}
where (\ref{w_low_form_8-_app}) holds if $m$ is sufficiently large, (\ref{w_low_form_8_app}) comes from the fact that $k=1$ and $h<1$, and (\ref{w_low_form_9_app}) comes from (\ref{w_low_form_0_app}). By (\ref{w_low_form_7_app}) and (\ref{w_low_form_10_app}), we have that, for a sufficiently large $m$, ${\sf NPoA}(\CG(m,s_m))\geq M$, thus showing part (ii) of the claim. 

If $h=0$, we consider a load balancing game defined as $\CG(m,s_m)$, but restricted to the resources of groups $R_1,\ldots, R_m$ and to the players of groups $N_1,\ldots, N_{m-1}$. By using the same proof arguments as those used for $h>0$, one can show the claim as well.

We now show part (i). Assume that $\mathcal{C}$ is abscissa-scaling and ordinate-scaling. Analogously to the proof of part (ii), we have that (\ref{w_low_form_0_app}) holds. Moreover, let $\CG'(m,s)$ be a weighted load balancing game equal to game $\CG(m,s)$ defined in the proof of part (ii), except for the strategy set of each player: for any $j\in [2m-1]$, the strategy set of each player of group $\N_j$ is ${\Sigma}_j:=R_j\cup R_{j+1}$. 
Let $\sg$ and $\sg^*$ be the strategy profiles defined as in game $\CG(m,s)$. 
By considering the case $h=j+1$ analyzed in the proof of part (ii) of the claim, it also holds that $\sg$ is a pure Nash equilibrium of $\CG'(m,s)$ for any $s\geq 1$. 
Therefore, if we take a sufficiently large $m$, an arbitrary $s\geq 1$, and by applying to game $\CG'(m,s)$ the same inequalities as in (\ref{w_low_form_7_app}) and (\ref{w_low_form_10_app}), part (i) follows.
\subsection{Proof of Lemma \ref{lemma_PoA_weighted_Polynomial latency functions}}
We have that
\begin{align}
&\sup_{k_1\geq o_1> 0,o_2>k_2\geq 0,f_1,f_2\in \mathcal{P}(p)}\left(\frac{f_1(k_1+o_1)}{f_1(o_1)}\right)^{\frac{(o_2-k_2)o_1}{k_1 o_2-k_2 o_1}}\left(\frac{f_2(k_2+o_2)}{f_2(o_2)}\right)^{\frac{(k_1-o_1)o_2}{k_1 o_2-k_2 o_1}}\nonumber\\
=&\sup_{\substack{k_1\geq o_1> 0,\\o_2>k_2\geq 0,\\\alpha_0,\ldots, \alpha_p,\geq 0\\\beta_0,\ldots,\beta_p\geq 0}}\left(\frac{\sum_{d=0}^p\alpha_d(k_1+o_1)^d}{\sum_{d=0}^p\alpha_d o_1^d}\right)^{\frac{(o_2-k_2)o_1}{k_1 o_2-k_2 o_1}}\left(\frac{\sum_{d=0}^p\beta_d(k_2+o_2)^d}{\sum_{d=0}^p\beta_d o_2^d}\right)^{\frac{(k_1-o_1)o_2}{k_1 o_2-k_2 o_1}}\nonumber\\
= &\sup_{\substack{k_1\geq o_1> 0,\\o_2>k_2\geq 0}}\left(\max_{d\in [p]\cup\{0\}}\frac{(k_1+o_1)^d}{o_1^d}\right)^{\frac{(o_2-k_2)o_1}{k_1 o_2-k_2 o_1}}\left(\max_{d\in [p]\cup\{0\}}\frac{(k_2+o_2)^d}{o_2^d}\right)^{\frac{(k_1-o_1)o_2}{k_1 o_2-k_2 o_1}}\nonumber\\
=&\sup_{k_1\geq o_1> 0,o_2>k_2\geq 0}\left(\left(\frac{k_1+o_1}{o_1}\right)^p\right)^{\frac{(o_2-k_2)o_1}{k_1 o_2-k_2 o_1}}\left(\left(\frac{k_2+o_2}{o_2}\right)^p\right)^{\frac{(k_1-o_1)o_2}{k_1 o_2-k_2 o_1}}\nonumber\\
=&\sup_{k\geq 1,0\leq h<1}\left((k+1)^{\frac{1-h}{k-h}}(h+1)^{\frac{k-1}{k-h}}\right)^p,\label{w_pol_form_11}
\end{align}
where (\ref{w_pol_form_11}) can be obtained  by setting $k:=k_1/o_1$ and $h:=k_2/o_2$. Now, we show that the maximum value of function $F(k,h):=(k+1)^{\frac{1-h}{k-h}}(h+1)^{\frac{k-1}{k-h}}$ over $k\geq 1$ and $0\leq h<1$ is equal to $2$. Observe that $
\ln(F(k,h))=\frac{1-h}{k-h}\ln(k+1)+\frac{k-1}{k-h}\ln(h+1)\leq \ln\left(\frac{1-h}{k-h}(k+1)+\frac{k-1}{k-h}(h+1)\right)$, 
where the last inequality holds since $\ln(F(k,h))$ is defined as convex combination of $\ln(k+1)$ and $\ln(h+1)$, and because of the concavity of the natural logarithm. Thus, we get
\begin{align}
F(k,h)&\leq \frac{1-h}{k-h}(k+1)+\frac{k-1}{k-h}(h+1)=\frac{(k-h)+(k-h)}{k-h}=2.\label{w_pol_form_13}
\end{align}
Finally, since $F(k,h)=2$ for $k=1$ and $h=0$, and because of (\ref{w_pol_form_13}), we have that the maximum of $F(k,h)$ over $k\geq 1$ and $0\leq h<1$ is $2$. Thus, we get that (\ref{w_pol_form_11}) is at most $2^p$. 
\subsection{Proof of Corollary \ref{cor1b}}
Let $\epsilon>0$. Let $\CG'(m)$ be the load balancing game defined as the game $\CG'(m,s)$ considered in the proof of part (i) of Theorem \ref{thm_w_low}, with $s=2$, $k=1$, $h=0$, and $f,g$ defined as $f(x)=g(x)=x^p$. One can easily observe that $\CG'(m)$ is a game with identical resources. Furthermore, because of the proof of Theorem \ref{thm_w_low}, there exists a sufficiently large integer $m$ such that ${\sf NPoA}(\CG'(m))>2^p-\epsilon$, and the claim follows by the arbitrariness of $\epsilon>0$. 
\section{Missing Proofs of Subsection 3.2}
\subsection{Proof of Theorem \ref{thm_unw_upp}}
Let $\CG\in {\sf ULB(\mathcal{C})}$ be an unweighted load balancing game with latency functions in $\mathcal{C}$, and let $\sg$ and $\sg^*$ be a worst-case pure Nash equilibrium and an optimal strategy profile of $\CG$, respectively. Let $k_j$ denote $k_j(\sg)$ and $o_j$ denote $k_j(\sg^*)$. As in Theorem \ref{thm_w_upp}, we get
\begin{align}
\prod_{j\in R(\sg)}\ell_{j}(k_{j})^{k_j}\leq \prod_{j\in R(\sg^*)}\ell_{j}(k_{j}+1)^{o_j}.\label{unw_form_nash} 
\end{align}
By exploiting the properties of the logarithmic function, we get
\begin{align}
\ln\left({\sf NPoA}(\CG)\right)&=\ln\left(\frac{\left({\prod_{j\in R(\sg)}\ell_j(k_{j})^{k_{j}}}\right)^{\frac{1}{n}}}{\left({\prod_{j\in R(\sg^*)}\ell_j(o_{j})^{o_{j}}}\right)^{\frac{1}{n}}}\right)\nonumber\\
&\leq \ln\left(\frac{\left({\prod_{j\in R(\sg^*)}\ell_j(k_{j}+1)^{o_{j}}}\right)^{\frac{1}{n}}}{\left({\prod_{j\in R(\sg^*)}\ell_j(o_{j})^{o_{j}}}\right)^{\frac{1}{n}}}\right)\label{unw_form_3b}\\
&=\frac{\sum_{j\in R(\sg^*)}o_j(\ln(\ell_j(k_j+1))-\ln(\ell_j(o_j)))}{\sum_{j\in R}k_j},\nonumber
\end{align}
where (\ref{unw_form_3b}) comes from (\ref{unw_form_nash}). 
Now, let $R_+:=\{j\in R(\sg^*):k_j\geq o_j \}$. 
%Continuing from (\ref{unw_form_4}), 
We have that
\begin{align}
&\frac{\sum_{j\in R(\sg^*)}o_j(\ln(\ell_j(k_j+1))-\ln(\ell_j(o_j)))}{\sum_{j\in R}k_j}\nonumber\\
\leq & \frac{\sum_{j\in R(\sg^*)}o_j(\ln(\ell_j(k_j+1))-\ln(\ell_j(o_j)))}{\sum_{j\in R(\sg^*)}k_j}\nonumber\\
\leq & \frac{\sum_{j\in R_+}o_j(\ln(\ell_j(k_j+1))-\ln(\ell_j(o_j)))}{\sum_{j\in R_+}k_j}\label{eqn_R_piu}\\
\leq & \max_{j\in R_+}\frac{o_j(\ln(\ell_j(k_j+1))-\ln(\ell_j(o_j)))}{k_j}\nonumber\\
\leq &\sup_{f\in\mathcal{C}, k\in \ZP, o\in [k]}\frac{o(\ln(f(k+1))-\ln(f(o)))}{k},\nonumber
\end{align}
where (\ref{eqn_R_piu}) holds because for any $j \in R(\sg^*) \setminus R_+$, it holds that $o_j(\ln(\ell_j(k_j+1))-\ln(\ell_j(o_j))) \leq 0$.
Therefore, we conclude that
\begin{equation*}
\ln\left({\sf NPoA}(\CG)\right)\leq \sup_{f\in\mathcal{C}, k\in \ZP, o\in  [k]}\frac{o(\ln(f(k+1))-\ln(f(o)))}{k},
\end{equation*}
and by exponentiating the previous inequality we get the claim. 
\subsection{Tightness of the Upper Bound in Theorem \ref{thm_unw_upp}}
\begin{theorem}\label{thm_unw_low}
Let $\mathcal{C}$ be a class of latency functions. If $\mathcal{C}$ is ordinate-scaling, then
$
{\sf NPoA}({\sf ULB}(\mathcal{C}))\geq \sup_{f\in\mathcal{C}, k\in \ZP, o\in  [k]}\left(\frac{f(k+1)}{f(o)}\right)^{\frac{o}{k}}.
$
\end{theorem}
\begin{proof}
In order to prove the theorem, we equivalently show that, for any $M<\sup_{f\in\mathcal{C}, k\in \ZP, o\in  [k]}\left(\frac{f(k+1)}{f(o)}\right)^{\frac{o}{k}}$, there exists a game $\CG\in {\sf ULB(\mathcal{C})}$ such that ${\sf NPoA}(\CG)> M$. 

Fix an arbitrary $M<\sup_{f\in\mathcal{C}, k\in \ZP, o\in  [k]}\left(\frac{f(k+1)}{f(o)}\right)^{\frac{o}{k}}$. Let $f\in\mathcal{C}$, $k\in \ZP$, $o\in  [k]$, and a sufficiently small $\epsilon>0$ such that
\begin{equation}\label{unw_low_form_0_app}
\left(\frac{f(k+1)}{f(o)}\right)^{\frac{o}{k}}> M+\epsilon.
\end{equation}

Given an integer $m>0$, let $\CG(m)$ be an unweighted load balancing game with $(k-o+1)m+o$ resources, partitioned into $m$ groups $R_1,R_2,\ldots, R_m$ such that $R_j:=\{r_{j,0},r_{j,1},\ldots, r_{j,k-o}\}$ for any $j\in [m-1]$, and $R_m:=\{r_{m,0},r_{m,1},\ldots, r_{m,k}\}$. Each resource $r_{j,h}$ has latency function $\ell_{r_{j,h}}(x):=\alpha_{j,h}f(x)$, with 
\begin{equation*}
\alpha_{j,h}:=
\begin{cases}
\left(\frac{f(k)}{f(k+1)}\right)^{j-1} & \text{ if }h=0\\
\frac{f(k)}{f(1)}\left(\frac{f(k)}{f(k+1)}\right)^{j-1} & \text{ otherwise.}
\end{cases}
\end{equation*}
We have $n:=mk$ players split into $m$ groups $\N_1,\N_2,\ldots ,\N_m$ of $k$ players each. For $j\in [m-1]$, the set of strategies ${\Sigma}_j$ of players of group $N_j$ is $R_j\cup \{r_{j+1,0}\}$, and the set of strategies ${\Sigma}_m$ of players in $\N_m$ is $R_m$. 

Let $\sg$ be the strategy profile such that, for any $j\in [m]$, all $k$ players of group $N_j$ select resource $r_{j,0}$, so that each resource $r_{j,0}$ has congestion $k$, and all the remaining resources have null congestion (see Figure \ref{fig:boat1}.a). 
We show that $\sg$ is a pure Nash equilibrium. Given an arbitrary player $i$ of group $\N_j$ with $j\in [m]$, such player has a cost equal to $\ell_{r_{j,0}}(k)=\alpha_{j,0}f(k)=\left(\frac{f(k)}{f(k+1)}\right)^{j-1}f(k)$ when playing strategy $\sigma_i$. If $j\in [m-1]$, and player $i$ unilaterally deviates to strategy $r_{j+1,0}$, her cost is $\ell_{r_{j+1,0}}(k+1)=\alpha_{j+1,0}f(k+1)=\left(\frac{f(k)}{f(k+1)}\right)^{j}f(k+1)=\left(\frac{f(k)}{f(k+1)}\right)^{j-1}f(k)=\ell_{r_{j,0}}(k)$, thus her cost does not improve. Analogously, if $j\in [m]$, and player $i$ unilaterally deviates to any strategy $r_{j,h}$ with $h\neq 0$, her cost is $\ell_{r_{j,h}}(1)=\alpha_{j,h}f(1)=\frac{f(k)}{f(1)}\left(\frac{f(k)}{f(k+1)}\right)^{j-1}f(1)=\ell_{r_{j,0}}(k)$, thus her cost does not improve as well. We conclude that $\sg$ is a pure Nash equilibrium. 

Now, let $\bm\sigma^*$ be a strategy profile defined as follows: (i) for any $j\in [m-1]$, $o$ players of group $N_j$ select resource $r_{j+1,0}$, and each of the $k-o$ remaining players of $N_j$ selects a distinct resource of $R_{j}\setminus\{r_{j,0}\}$, (ii) all the $k$ players of group $N_m$ select a distinct resource of $E_{m}\setminus\{r_{m,0}\}$. Thus, in $\bm\sigma^*$, any resource of type $r_{j,0}$ with $j>1$ has congestion $o$, resource $r_{1,0}$ has null congestion, and the remaining resources have unitary congestion (see Figure \ref{fig:boat1}.b). 
\begin{figure}
  \includegraphics[width=\linewidth]{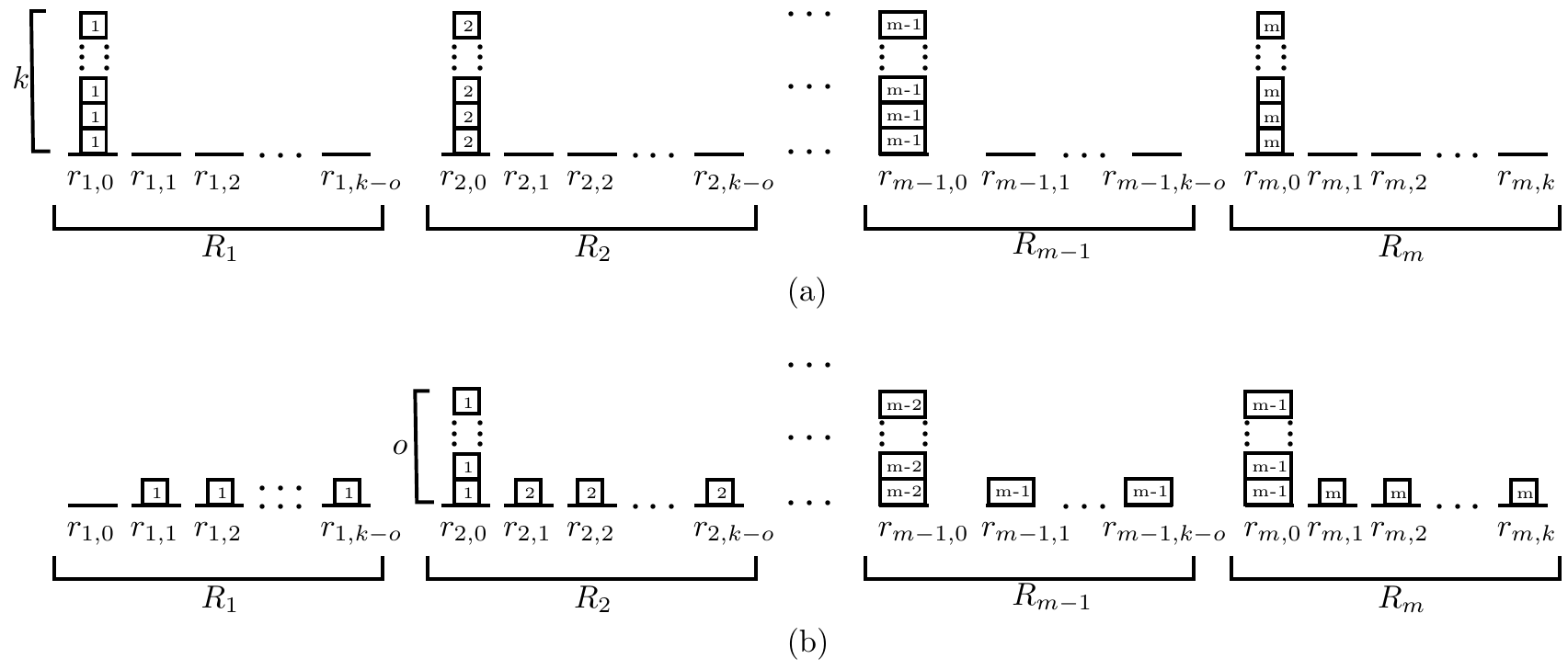}
  \caption{The $\CG$ used in the proof of Theorem \ref{thm_unw_low}. Columns represent resources and squares represent players (number $j$ inside a square means that the player belongs to group $N_j$). (a): The Nash equilibrium $\bm\sigma$; (b): The strategy profile ${\bm\sigma}^*$.}
  \label{fig:boat1}
\end{figure}
%&\ {\sf NPoA}(\CG(m))\nonumber\\
By some algebraic manipulation, it holds that
\begin{align}
&\ \frac{{\sf NSW}(\sg)}{{\sf NSW}(\sg^*)}=\nonumber\\
&=\left(\frac{\prod_{j=1}^m\ell_{j,0}(k)^k}{\prod_{j=1}^{m-1}\left(\ell_{r_{j+1,0}}(o)^o\prod_{r\in R_{j}\setminus\{r_{j,0}\}}\ell_{r}(1)\right)\prod_{r\in R_{m}\setminus\{r_{m,0}\}}\ell_{r}(1)}\right)^{\frac{1}{km}}\nonumber\\
&=\left(\frac{\prod_{j=1}^m\left(\left(\frac{f(k)}{f(k+1)}\right)^{j-1}f(k)\right)^k}{\prod_{j=1}^{m-1}\left[\left(\left(\frac{f(k)}{f(k+1)}\right)^{j}f(o)\right)^o\left(\frac{f(k)}{f(1)}\left(\frac{f(k)}{f(k+1)}\right)^{j-1}f(1)\right)^{k-o}\right]\left(\frac{f(k)}{f(1)}\left(\frac{f(k)}{f(k+1)}\right)^{m-1}f(1)\right)^{k}}\right)^{\frac{1}{km}}\nonumber\\
&=\left(\frac{\prod_{j=1}^m\left(\left(\frac{f(k)}{f(k+1)}\right)^{j}f(k+1)\right)^k}{\prod_{j=1}^{m-1}\left[\left(\left(\frac{f(k)}{f(k+1)}\right)^{j}f(o)\right)^o\left(\left(\frac{f(k)}{f(k+1)}\right)^{j}f(k+1)\right)^{k-o}\right]\left(\left(\frac{f(k)}{f(k+1)}\right)^{m}f(k+1)\right)^{k}}\right)^{\frac{1}{km}}\nonumber\\
&=\left(\frac{\left(\prod_{j=1}^m\left(\frac{f(k)}{f(k+1)}\right)^{kj}\right)f(k+1)^{km}}{\left(\prod_{j=1}^{m-1}\left(\frac{f(k)}{f(k+1)}\right)^{kj}\right)f(o)^{o(m-1)}f(k+1)^{(k-o)(m-1)}\left(\frac{f(k)}{f(k+1)}\right)^{km}f(k+1)^{k}}\right)^{\frac{1}{km}}\nonumber\\
&=\left(\frac{f(k+1)^{km}}{f(o)^{o(m-1)}f(k+1)^{(k-o)(m-1)}f(k+1)^{k}}\right)^{\frac{1}{km}}\nonumber\\
&=\left(\frac{f(k+1)^{o(m-1)}}{f(o)^{o(m-1)}}\right)^{\frac{1}{km}}\nonumber\\
&=\left(\frac{f(k+1)}{f(o)}\right)^{\frac{o(m-1)}{km}}.\label{unw_low_form_1_app}
\end{align}
By using (\ref{unw_low_form_0_app}) and (\ref{unw_low_form_1_app}), and by choosing a sufficiently large $m$, we get
\begin{align*}
{\sf NPoA}(\CG(m))&\geq \frac{{\sf NSW}(\sg)}{{\sf NSW}(\sg^*)}\\
&= \left(\frac{f(k+1)}{f(o)}\right)^{\frac{o(m-1)}{km}}\\
&\geq \lim_{m\rightarrow \infty}\left(\frac{f(k+1)}{f(o)}\right)^{\frac{o(m-1)}{km}}-\epsilon\\
&=\left(\frac{f(k+1)}{f(o)}\right)^{\frac{o}{k}}-\epsilon\\
&> M+\epsilon-\epsilon\\
&=M,
\end{align*}
thus showing the claim. \qed
\end{proof}
\subsection{Proof of Corollary \ref{cor2}}
The claim follows from the following lemma.
\begin{lemma}\label{lemma_PoA_unweighted_Polynomial latency functions}
$\sup_{f\in\mathcal{P}(p),  k\in \ZP, o\in  [k]}\left(\frac{f(k+1)}{f(o)}\right)^{\frac{o}{k}} = 2^p.$
%The Nash price of anarchy of weighted load balancing games with polynomial latency functions (even for symmetric games) of maximum degree $p$ is ${\sf NPoA}({\sf WLB}(\mathcal{P}(p)))={\sf NPoA}({\sf SWLB}(\mathcal{P}(p)))=2^p$. 
\end{lemma}
\begin{proof}
We have that
\begin{align}
\sup_{f\in\mathcal{P}(p),  k\in \ZP, o\in  [k]}\left(\frac{f(k+1)}{f(o)}\right)^{\frac{o}{k}} &=\sup_{\alpha_0,\alpha_1,\ldots, \alpha_p\geq 0,  k\in \ZP, o\in  [k]}\left(\frac{\sum_{d=0}^p\alpha_d (k+1)^d}{\sum_{d=0}^p\alpha_d o^d}\right)^{\frac{o}{k}} \nonumber\\
&= \sup_{k\in \ZP, o\in  [k]}\left(\max_{d\in [p]\cup\{0\}}\frac{(k+1)^d}{o^d}\right)^{\frac{o}{k}} \nonumber\\
&= \sup_{k\in \ZP, o\in  [k]} \left( \left( \frac{k+1}{o} \right)^p \right) ^{\frac{o}{k}} \nonumber\\
&=\left(\sup_{k\in \ZP, o\in  [k]}\left(\frac{k+1}{o}\right)^{\frac{o}{k}}\right)^p\nonumber\\
& = 2^p,\label{w_pol_form_0_12}
\end{align}
where (\ref{w_pol_form_0_12}) holds for the following reasons: First of all, we have that $(\frac{k+1}{o})^{\frac{o}{k}}=2$ if $o=k=1$, thus showing that $2\leq \sup_{k\in \ZP, o\in  [k]}\left(\frac{k+1}{o}\right)^{\frac{o}{k}}$; furthermore, by setting $x:=k/o$, we obtain $(\frac{k+1}{o})^{\frac{o}{k}} = \left(x+\frac{1}{o}\right)^{\frac{1}{x}} \leq (x+1)^{\frac{1}{x}} \leq 2$, where the last inequality is equivalent to the well-known inequality $2^x\geq x+1$ which holds for any $x\geq 1$.\qed
\end{proof}
\section{Missing Proofs of Subsection 3.3}
\subsection{Proof of Theorem \ref{thm_non_upp}}
Let $\NCG\in {\sf NLB(\mathcal{C})}$ be a non-atomic load balancing game with latency functions in $\mathcal{C}$, and let $\nsg$ and $\nsg^*$ be a worst-case pure Nash equilibrium and an optimal strategy profile of $\NCG$, respectively. Let $k_j$ denote $k_j(\nsg)$ and $o_j$ denote $k_j(\nsg^*)$. 

For any player type $i$ and pair $(j,j^*)$ of resources, let $\alpha^i_{j,j^*}$ be the amount of players of type $i$ selecting resource $j$ in $\nsg$ and resource $j^*$ in $\nsg^*$. Clearly, it holds that, for any $i \in N$, $\sum_{j,j^* \in R} \alpha^i_{j,j^*}= r_i$.  

Since $\nsg$ is a pure Nash equilibrium, if there exists $i \in N$ such that $\alpha^i_{j,j^*}>0$, we have that $cost_j(\nsg)\leq cost_{j^*}(\nsg)$. For any $j,j^* \in R$, let $A_{j,j^*}=\sum_{i \in N} \alpha^i_{j,j^*}$. Clearly, it holds that 
\begin{equation}\label{nonatomic_form_1}
cost_j(\nsg)^{A_{j,j^*}}\leq cost_{j^*}(\nsg)^{A_{j,j^*}}.
\end{equation}

Since, for any $j \in R(\nsg)$, $\sum_{j^* \in R} A_{j,j^*} = k_j$ and, symmetrically, for any $j^* \in R(\nsg^*)$, $\sum_{j \in R} A_{j,j^*} = o_j$, it follows that 
\begin{equation}\label{nonatomic_form_2}
\prod_{j,j^* \in R} cost_j(\nsg)^{A_{j,j^*}} = \prod_{j \in R(\nsg)} cost_j(\nsg)^{k_j}
\end{equation}
and
\begin{equation}\label{nonatomic_form_3}
\prod_{j,j^* \in R} cost_{j^*}(\nsg^*)^{A_{j,j^*}} = \prod_{j \in R(\nsg^*)} cost_j(\nsg)^{o_j}.
\end{equation}

By multiplying (\ref{nonatomic_form_1}) over all pairs of resources in $R$ and by exploiting (\ref{nonatomic_form_2}) and (\ref{nonatomic_form_3}), we obtain
\begin{align}
&\prod_{j\in R(\nsg)}\ell_{j}(k_{j})^{k_j}= \prod_{j \in R(\nsg)} cost_j(\nsg)^{k_j} 
=\prod_{j,j^* \in R}cost_j(\nsg)^{A_{j,j^*}}\nonumber\\
&\leq \prod_{j,j^* \in R}cost_{j^*}(\nsg)^{A_{j,j^*}}= \prod_{j \in R(\nsg^*)} cost_j(\nsg)^{o_j} 
= \prod_{j\in R(\nsg^*)}\ell_{j}(k_{j})^{o_j}.\label{nonatomic_form_nash} 
\end{align}

By exploiting the properties of the logarithmic function, we get
\begin{align}
\ln\left({\sf NPoA}(\CG)\right)&=\ln\left(\frac{\left({\prod_{j\in R(\nsg)}\ell_j(k_{j})^{k_{j}}}\right)^{\frac{1}{\sum_{i \in N} r_i}}}{\left({\prod_{j\in R(\nsg^*)}\ell_j(o_{j})^{o_{j}}}\right)^{\frac{1}{\sum_{i \in N} r_i}}}\right)\nonumber\\
&\leq \ln\left(\frac{\left({\prod_{j\in R(\nsg^*)}\ell_j(k_{j})^{o_{j}}}\right)^{\frac{1}{\sum_{i \in N} r_i}}}{\left({\prod_{j\in R(\nsg^*)}\ell_j(o_{j})^{o_{j}}}\right)^{\frac{1}{\sum_{i \in N} r_i}}}\right)\label{nonatomic_form_3b}\\
&=\frac{\sum_{j\in R(\nsg^*)}o_j\ln(\ell_j(k_j))-\sum_{j\in R(\nsg^*)}o_j\ln(\ell_j(o_j))}{\sum_{i \in N} r_i}\nonumber\\
&=\frac{\sum_{j\in R(\nsg^*)}o_j(\ln(\ell_j(k_j))-\ln(\ell_j(o_j)))}{\sum_{j \in R} k_j},\nonumber\\
&\leq \frac{\sum_{j\in R_+}o_j(\ln(\ell_j(k_j))-\ln(\ell_j(o_j)))}{\sum_{j\in R_+}k_j}\label{nonatomic_eqn_R_piu}\\
&\leq  \max_{j\in R_+}\frac{o_j(\ln(\ell_j(k_j))-\ln(\ell_j(o_j)))}{k_j}\nonumber\\
&\leq \sup_{f\in\mathcal{C}, k\geq o>0}\frac{o(\ln(f(k))-\ln(f(o)))}{k},\nonumber
\end{align}
where (\ref{nonatomic_form_3b}) comes from (\ref{nonatomic_form_nash}), and (\ref{nonatomic_eqn_R_piu}) is obtained by using similar arguments as in Theorem \ref{thm_unw_upp} (in particular, see inequalities (\ref{eqn_R_piu})). 
Therefore, we conclude that
\begin{equation*}
\ln\left({\sf NPoA}(\NCG)\right)\leq \sup_{f\in\mathcal{C}, k\geq o>0}\frac{o(\ln(f(k))-\ln(f(o)))}{k},
\end{equation*}
and by exponentiating the previous inequality we get the claim. 
\subsection{Tightness of the Upper Bound of Theorem \ref{thm_non_upp}}
\begin{theorem}\label{thm_non_low}
Let $\mathcal{C}$ be a class of latency functions. If $\mathcal{C}$ is all-constant-including, then
$
{\sf NPoA}({\sf NLB}(\mathcal{C}))=
{\sf NPoA}({\sf SNLB}(\mathcal{C}))\geq \sup_{f\in\mathcal{C},  k\geq o>0}\left(\frac{f(k)}{f(o)}\right)^{\frac{o}{k}}.
$
\end{theorem}
\begin{proof}
To show the theorem, we equivalently show that, for any $M<\sup_{f\in\mathcal{C}, k\geq o>0}\left(\frac{f(k)}{f(o)}\right)^{\frac{o}{k}}$, there exists a symmetric non-atomic load balancing game $\NCG\in {\sf SNLB(\mathcal{C})}$ such that ${\sf NPoA}(\NCG)> M$. 
Fix an arbitrary $M<\sup_{f\in\mathcal{C},  k\geq o>0}\left(\frac{f(k)}{f(o)}\right)^{\frac{o}{k}}$. Let $f\in\mathcal{C}$ and $k\geq o>0$ such that $
\left(\frac{f(k)}{f(o)}\right)^{\frac{o}{k}}> M$. Let $\NCG$ be a symmetric non-atomic load balancing game with a unique player type, say $1$, and two resources having latency defined as $\ell_{1}(x):=f(x)$ and $\ell_{2}(x):=f(k)$. 
Assume that the amount of players of type $1$ is $r_1=k$. 
Let $\nsg$ be the strategy profile in which all players select resource $1$, and let $\nsg^*$ be the strategy profile in which an amount $o$ of players selects resource $1$ and the remaining one (i.e., $k-o$) selects resource $2$. 
We trivially have that $\nsg$ is a pure Nash equilibrium. 
Thus, we obtain ${\sf NPoA}(\NCG)\geq \frac{{\sf NSW}(\nsg)}{{\sf NSW}(\nsg^*)}=\left(\frac{\ell_{1}(k)^{k}}{\ell_{1}(o)^{o}\ell_{2}(k-o)^{k-o}}\right)^{\frac{1}{k}}=\left(\frac{f(k)^{k}}{f(o)^{o}f(k)^{k-o}}\right)^{\frac{1}{k}}=\left(\frac{f(k)}{f(o)}\right)^{\frac{o}{k}}>M,
$ and the claim follows. \qed
\end{proof}
\subsection{Proof of Corollary \ref{cor3}}
We have that 
\begin{align}
{\sf NPoA}({\sf NLB}(\mathcal{P}(p)))&={\sf NPoA}({\sf SNLB}(\mathcal{P}(p)))\label{non_pol_form_-1}\\
&=\sup_{f\in\mathcal{C}, k\geq o>0}\left(\frac{f(k)}{f(o)}\right)^{\frac{o}{k}}\label{non_pol_form_0}\\
&=\sup_{\alpha_0,\alpha_1,\ldots, \alpha_p\geq 0, k\geq o>0}\left(\frac{\sum_{d=0}^p\alpha_d k^d}{\sum_{d=0}^p\alpha_d o^d}\right)^{\frac{o}{k}}\nonumber\\
&= \sup_{k\geq o>0}\left(\max_{d\in [p]\cup\{0\}}\frac{k^d}{o^d}\right)^{\frac{o}{k}}\nonumber\\
&=\max_{d\in [p]\cup\{0\}}\sup_{k\geq o>0}\left(\frac{k^d}{o^d}\right)^{\frac{o}{k}}\nonumber\\
&=\max_{d\in [p]\cup\{0\}}\left(\sup_{k\geq o>0}\left(\frac{k}{o}\right)^{\frac{o}{k}}\right)^d\nonumber\\
&=\left(\sup_{k\geq o>0}\left(\frac{k}{o}\right)^{\frac{o}{k}}\right)^p,\nonumber\\
&=\left(\sup_{x>0}x^{\frac{1}{x}}\right)^p,\label{non_pol_form_1}\\
&=\left(e^{\frac{1}{e}}\right)^p,\label{non_pol_form_2}
\end{align}
where (\ref{non_pol_form_-1}) and (\ref{non_pol_form_0}) come from Theorems \ref{thm_non_upp} and \ref{thm_non_low} (observe that polynomial latency functions are all-constant-including), (\ref{non_pol_form_1}) can be obtained  by setting $x:=k/o$, and (\ref{non_pol_form_2}) comes from the fact that function $F(x):=x^{1/x}$ is maximized by $x=e$. 
\section{Missing Proofs of Section 4}
\subsection{Proof of Theorem \ref{thm_onl_upp}}
Let $\I\in {\sf WLB}(\mathcal{C})$ be a load balancing instance with latency functions in $\mathcal{C}$, and let $\sg$ and $\sg^*$ be the states returned by the greedy algorithm and an optimal strategy profile of $\CG$, respectively.

Let $k_j$ denote $k_j(\sg)$ and $o_j$ denote $k_j(\sg^*)$. For any $i\in N$ and resource $j$, let $(\sg^{i})$ be the partial state in which the first $i$ clients have been assigned according to $\sg$, and let $(\sg^{i-1},j)$ be the state in which the first $i-1$ clients have been assigned according to $\sg$ and client $i$ is assigned to resource $j$. By definition of greedy algorithm, we have that $\sigma_i\in \arg\min_{j\in R}{\sf NSW}(\sg^{i-1},j)=\arg\min_{j\in R}\frac{\prod_{l\leq i}cost_l(\sg^{i-1},j)}{\prod_{l\leq i-1}cost_l(\sg^{i-1})}=\arg\min_{j\in R}\frac{\ell_j(k_j(\sg^{i-1},j))^{k_j(\sg^{i-1},j)}}{\ell_j(k_j(\sg^{i-1}))^{k_j(\sg^{i-1})}}$, where we set $\ell_j(0)^{0}:=1$. Thus, we can equivalently define the greedy assignment by saying that each client $i$ is assigned to the resource $j$ minimizing $\frac{\ell_j(k_j(\sg^{i-1},j))^{k_j(\sg^{i-1},j)}}{\ell_j(k_j(\sg^{i-1}))^{k_j(\sg^{i-1})}}$, so that
\begin{equation}\label{onl_form_0}
\frac{\ell_{\sigma_i}(k_{\sigma_i}(\sg^{i}))^{k_{\sigma_i}(\sg^{i})}}{\ell_{\sigma_i}(k_{\sigma_i}(\sg^{i-1}))^{k_{\sigma_i}(\sg^{i-1})}}\leq \frac{\ell_{\sigma_i^*}(k_{\sigma_i}(\sg^{i-1},\sigma_i^*))^{k_{\sigma_i^*}(\sg^{i-1},\sigma_i^*)}}{\ell_{\sigma_i^*}(k_{\sigma_i^*}(\sg^{i-1}))^{k_{\sigma_i^*}(\sg^{i-1})}}.
\end{equation}
We have that:
\begin{align}
\prod_{i\in \N}\frac{\ell_{\sigma_i}(k_{\sigma_i}(\sg^{i}))^{k_{\sigma_i}(\sg^{i})}}{\ell_{\sigma_i}(k_{\sigma_i}(\sg^{i-1}))^{k_{\sigma_i}(\sg^{i-1})}}&=\prod_{j\in R(\sg)}\prod_{i\in \N:\sigma_i=j}\frac{\ell_{j}(k_{j}(\sg^{i}))^{k_{j}(\sg^{i})}}{\ell_{j}(k_{j}(\sg^{i-1}))^{k_{j}(\sg^{i-1})}}\nonumber\\
&=\prod_{j\in R(\sg)}\ell_{j}(k_{j}(\sg^{n}))^{k_{j}(\sg^{n})}\label{onl_form_1-}\\
&=\prod_{j\in R(\sg)}\ell_{j}(k_{j})^{k_{j}},\label{onl_form_1}
\end{align}
where (\ref{onl_form_1-}) is obtained by exploiting telescoping properties. Furthermore, we get
\begin{align}
&\prod_{i\in \N}\frac{\ell_{\sigma_i^*}(k_{\sigma_i^*}(\sg^{i-1},\sigma_i^*))^{k_{\sigma_i^*}(\sg^{i-1},\sigma_i^*)}}{\ell_{\sigma_i^*}(k_{\sigma_i^*}(\sg^{i-1}))^{k_{\sigma_i^*}(\sg^{i-1})}}\nonumber\\
&=\prod_{i\in \N}\frac{\ell_{\sigma_i^*}(k_{\sigma_i^*}(\sg^{i-1})+w_i)^{k_{\sigma_i^*}(\sg^{i-1})+w_i}}{\ell_{\sigma_i^*}(k_{\sigma_i^*}(\sg^{i-1}))^{k_{\sigma_i^*}(\sg^{i-1})}}\nonumber\\
&\leq \prod_{i\in \N}\frac{\ell_{\sigma_i^*}(k_{\sigma_i^*}+w_i)^{k_{\sigma_i^*}+w_i}}{\ell_{\sigma_i^*}(k_{\sigma_i^*})^{k_{\sigma_i^*}}}\label{onl_form_2---}\\
&= \prod_{j\in R(\sg^*)}\prod_{i\in \N:\sigma_i^*=j}\frac{\ell_{j}(k_{j}+w_i)^{k_{j}+w_i}}{\ell_{j}(k_{j})^{k_{j}}}\nonumber\\
&\leq \prod_{j\in R(\sg^*)}\prod_{i\in \N:\sigma_i^*=j}\frac{\ell_{j}(k_{j}+\sum_{t\leq i:\sigma^*_t=j}w_t)^{k_{j}+\sum_{t\leq i:\sigma^*_t=j}w_t}}{\ell_{j}(k_{j}+\sum_{t< i:\sigma^*_t=j}w_t)^{k_{j}+\sum_{t< i:\sigma^*_t=j}w_t}}\label{onl_form_2--}\\
&= \prod_{j\in R(\sg^*)}\frac{\ell_{j}(k_{j}+\sum_{t:\sigma^*_t=j}w_t)^{k_{j}+\sum_{t:\sigma^*_t=j}w_t}}{\ell_{j}(k_{j})^{k_{j}}}\label{onl_form_2-}\\
&= \prod_{j\in R(\sg^*)}\frac{\ell_{j}(k_{j}+o_{j})^{k_{j}+o_{j}}}{\ell_{j}(k_{j})^{k_{j}}}\label{onl_form_2},
\end{align}
where (\ref{onl_form_2-}) is obtained by exploiting telescoping properties, and (\ref{onl_form_2---}) and (\ref{onl_form_2--}) easily come from the following fact:
\begin{fact}
Given a quasi-log-convex latency function $f$, we have that $\frac{f(x+z)^{x+z}}{f(x)^{}}\leq \frac{f(x+y+z)^{x+y+z}}{f(x+y)^{x+y}}$ for any $x,y,z\geq 0$.
\end{fact}
\begin{proof}
Since the function $g$ such that $g(t)=t\ln(f(t))$ is convex, we have that $g(x+z)-g(x)\leq g(x+y+z)-g(x+y)$ for any $x,y,z\geq 0$, thus, by exponentiating the previous inequality, the claim follows. \qed
\end{proof}
By putting together (\ref{onl_form_0}), (\ref{onl_form_1}), and (\ref{onl_form_2}), we get
\begin{align}
\prod_{j\in R(\sg)}\ell_{j}(k_j)^{k_j}&= \prod_{i\in \N}\frac{\ell_{\sigma_i}(k_{\sigma_i}(\sg^{i}))^{k_{\sigma_i}(\sg^{i})}}{\ell_{\sigma_i}(k_{\sigma_i}(\sg^{i-1}))^{k_{\sigma_i}(\sg^{i-1})}}\nonumber\\
&\leq \prod_{i\in \N}\frac{\ell_{\sigma_i^*}(k_{\sigma_i^*}(\sg^{i-1},\sigma_i^*))^{k_{\sigma_i^*}(\sg^{i-1},\sigma_i^*)}}{\ell_{\sigma_i^*}(k_{\sigma_i^*}(\sg^{i-1}))^{k_{\sigma_i^*}(\sg^{i-1})}}\nonumber\\
&\leq \prod_{j\in R(\sg^*)}\frac{\ell_{j}(k_{j}+o_{j})^{k_{j}+o_{j}}}{\ell_{j}(k_{j})^{k_{j}}}.\label{onl_form_nash} 
\end{align}
By exploiting the properties of the logarithmic function, we obtain
\begin{align}
&\ln\left({\sf CR}_{\sf G}(\I)\right)\nonumber\\
&=\ln\left(\frac{\left({\prod_{j\in R(\sg)}\ell_j(k_{j})^{k_{j}}}\right)^{\frac{1}{\sum_{i\in N}w_i}}}{\left({\prod_{j\in R(\sg^*)}\ell_j(o_{j})^{o_{j}}}\right)^{\frac{1}{\sum_{i\in N}w_i}}}\right)\nonumber\\
&\leq \ln\left(\frac{\left(\prod_{j\in R(\sg^*)}\frac{\ell_{j}(k_{j}+o_{j})^{k_{j}+o_{j}}}{\ell_{j}(k_{j})^{k_{j}}}\right)^{\frac{1}{\sum_{i\in N}w_i}}}{\left({\prod_{j\in R(\sg^*)}\ell_j(o_{j})^{o_{j}}}\right)^{\frac{1}{\sum_{i\in N}w_i}}}\right)\label{onl_form_3b}\\
&=\frac{\sum_{j\in R(\sg^*)}\left((k_j+o_j)\ln(\ell_j(k_j+o_j))-k_j\ln(\ell_j(k_j))-o_j\ln(\ell_j(o_j))\right)}{\sum_{i\in N}w_i},\label{onl_form_4b}
\end{align}
where (\ref{onl_form_3b}) comes from (\ref{onl_form_nash}). Since $\sum_{i\in \N}w_i=\sum_{j\in R}k_j=\sum_{j\in R}o_j$, we have that  (\ref{onl_form_4b}) is upper bounded by the optimal solution of the following optimization problem on some new linear variables $(\alpha_j)_{j\in R}$:
\begin{align*}
\max \quad & \frac{\sum_{j\in R(\sg^*)}\alpha_j\left((k_j+o_j)\ln(\ell_j(k_j+o_j))-k_j\ln(\ell_j(k_j))-o_j\ln(\ell_j(o_j))\right)}{\sum_{j\in R}\alpha_j k_j}\\
\text{s.t.} \quad &\sum_{j\in R}\alpha_j k_j=\sum_{j\in R}\alpha_j o_j,\quad  \alpha_j\geq 0\ \forall j\in R.
\end{align*}
By normalizing the denominator of the objective function, we obtain the following equivalent linear program:
\begin{align}
\max \quad & \sum_{j\in R(\sg^*)}\alpha_j\left((k_j+o_j)\ln(\ell_j(k_j+o_j))-k_j\ln(\ell_j(k_j))-o_j\ln(\ell_j(o_j))\right)\label{onl_form_5b}\\
\text{s.t.} \quad &\sum_{j\in R}\alpha_j k_j=1,\quad   \sum_{j\in R}\alpha_j o_j=1,\quad \alpha_j\geq 0\ \forall j\in R.\nonumber
\end{align}
We have the following fact, whose proof is omitted, since it is similar to that of Fact \ref{prel_lem_w}. 
\begin{fact}\label{prel_lem_onl}
The maximum value of the linear program considered in (\ref{w_form_5b}) is at most $$\sup_{k_1\geq  o_1> 0,o_2>k_2\geq 0,f_1,f_2\in \mathcal{C}}\frac{(o_2-k_2) F(f_1,o_1,k_1)+(k_1-o_1)F(f_2,o_2,k_2)}{k_1 o_2-k_2 o_1},$$ where $F(f,o,k):=(k+o)\ln(f(k+o))-k\ln(f(k))-o\ln(f(o)).$ 
\end{fact}
By continuing from (\ref{onl_form_4b}) and by using Fact \ref{prel_lem_onl}, we get
\begin{equation*}
\ln\left({\sf CR}_{\sf G}(\I)\right)\leq \sup_{k_1\geq  o_1> 0,o_2>k_2\geq 0,f_1,f_2\in \mathcal{C}}\frac{(o_2-k_2) F(f_1,o_1,k_1)+(k_1-o_1)F(f_2,o_2,k_2)}{k_1 o_2-k_2 o_1}.
\end{equation*}
By exponentiating the previous inequality, we get the claim. 
\subsection{Tightness of the Upper Bound of Theorem \ref{thm_onl_upp}}
\begin{theorem}\label{thm_onl_low}
Let $\mathcal{C}$ be a class of latency functions and let ${\sf G}$ be the greedy algorithm. If $\mathcal{C}$ is abscissa-scaling and ordinate-scaling, then 
\begin{align}
&{\sf CR}_{\sf G}({\sf WLB}(\mathcal{C}))\nonumber\\
&\geq \sup_{\substack{k_1\geq  o_1> 0,\\o_2>k_2\geq 0,\\f_1,f_2\in \mathcal{C}}}\left(\frac{f_1(k_1+o_1)^{k_1+o_1}}{f_1(k_1)^{k_1}f_1(o_1)^{o_1}}\right)^{\frac{o_2-k_2}{o_2 k_1-o_1k_2}}\left(\frac{f_2(k_2+o_2)^{k_2+o_2}}{f_2(k_2)^{k_2}f_2(o_2)^{o_2}}\right)^{\frac{k_1-o_1}{o_2 k_1-o_1k_2}}.\label{low_onl_weig1}
\end{align}
\end{theorem}
\begin{proof}
Let us assume that $\mathcal{C}$ is abscissa-scaling and ordinate-scaling. We equivalently show that for any $M<\sup_{k_1\geq  o_1> 0,o_2>k_2\geq 0,f_1,f_2\in \mathcal{C}}\left(\frac{f_1(k_1+o_1)^{k_1+o_1}}{f_1(k_1)^{k_1}f_1(o_1)^{o_1}}\right)^{\frac{o_2-k_2}{o_2 k_1-o_1k_2}}\left(\frac{f_2(k_2+o_2)^{k_2+o_2}}{f_2(k_2)^{k_2}f_2(o_2)^{o_2}}\right)^{\frac{k_1-o_1}{o_2 k_1-o_1k_2}}$ there exists an instance $\I\in {\sf WLB}(\mathcal{C})$ such that ${\sf NPoA}(\I)>M$. 

Let $f_1,f_2\in\mathcal{C}$, $k_1,k_2,o_1,o_2\geq 0$ such that $k_1\geq o_1>0,o_2>k_2\geq 0$, and let $\epsilon>0$ be a sufficiently small number such that $\left(\frac{f_1(k_1+o_1)^{k_1+o_1}}{f_1(k_1)^{k_1}f_1(o_1)^{o_1}}\right)^{\frac{o_2-k_2}{o_2 k_1-o_1k_2}}\left(\frac{f_2(k_2+o_2)^{k_2+o_2}}{f_2(k_2)^{k_2}f_2(o_2)^{o_2}}\right)^{\frac{k_1-o_1}{o_2 k_1-o_1k_2}}> M+\epsilon.$ 
Let $f,g\in\mathcal{C}$ be such that $f(x):=f_1(o_1x)$ and $g(x):=f_2(o_2x)$,  and let $k:=k_1/o_1$ and $h:=k_2/o_2$. Since $$\left(\frac{f_1(k_1+o_1)^{k_1+o_1}}{f_1(k_1)^{k_1}f_1(o_1)^{o_1}}\right)^{\frac{o_2-k_2}{o_2 k_1-o_1k_2}}\left(\frac{f_2(k_2+o_2)^{k_2+o_2}}{f_2(k_2)^{k_2}f_2(o_2)^{o_2}}\right)^{\frac{k_1-o_1}{o_2 k_1-o_1k_2}}=\left(\frac{f(k+1)^{k+1}}{f(k)^{k}f(1)}\right)^{\frac{1-h}{k-h}}\left(\frac{g(h+1)^{h+1}}{g(h)^{h}g(1)}\right)^{\frac{k-1}{k-h}}$$

we have that 
%$\left(\frac{f(k+1)}{f(1)}\right)^{\frac{1}{k}}\geq M+\epsilon$, with $f\in\mathcal{C}$ and $k\geq 1$. %.Thus we can assume that 
\begin{equation}\label{onl_low_form_0}
\left(\frac{f(k+1)^{k+1}}{f(k)^{k}f(1)}\right)^{\frac{1-h}{k-h}}\left(\frac{g(h+1)^{h+1}}{g(h)^{h}g(1)}\right)^{\frac{k-1}{k-h}}> M+\epsilon\text{, for some } f,g\in\mathcal{C},\ k\geq 1,\text{ and }h<1.
\end{equation}
First of all, we assume that $h>0$. Given an integer $m\geq 3$, let $\I(m)$ be a 
load balancing instance having $2m$ resources $r_1,r_2,r_3\ldots, r_{2m}$ and $2m-1$ clients such that the set of strategies of each client $j$ is $\{r_j,r_{j+1}\}$. Each resource $r_j$ has a latency function defined as $\ell_j(x):=\alpha_j \hat{f}_j\left(\beta_j x\right)$, and the weight of each client $j$ is defined as $w_j:=1/\beta_{j+1}$, where $\alpha_j, \hat{f}_j$, and $\beta_j$ are defined as follows:
\begin{align}
&\hat{f}_j:=
\begin{cases}
f & \text{ if }j\leq m-1\\
g & \text{ if }j\geq m
\end{cases},\quad 
\beta_j:=
\begin{cases}
\left(\frac{1}{k}\right)^{j-1} & \text{ if }j\leq m-1\\
\left(\frac{1}{h}\right)^{j-m}\left(\frac{1}{k}\right)^{m-1} & \text{ if }m\leq j\leq 2m
\end{cases},\\
&\alpha_j:=
\begin{cases}
\left(\frac{f(k)^{k+1}}{f(k+1)^{k+1}}\right)^{j-1} & \text{ if }j\leq m-1\\
\left(\frac{g(h)^{h+1}}{g(h+1)^{h+1}}\right)^{j-m}\left(\frac{f(k)g(h)^h}{g(h+1)^{h+1}}\right)\left(\frac{f(k)^{k+1}}{f(k+1)^{k+1}}\right)^{m-2} & \text{ if }m\leq j\leq 2m-1\\
\frac{g(h)}{g(1)}\left(\frac{g(h)^{h+1}}{g(h+1)^{h+1}}\right)^{m-1}\left(\frac{f(k)g(h)^h}{g(h+1)^{h+1}}\right)\left(\frac{f(k)^{k+1}}{f(k+1)^{k+1}}\right)^{m-2} & \text{ if }j=2m
\end{cases}.
\end{align}
Observe that, by construction of $\alpha_j,\beta_j,w_j$, the following properties hold:\small
\begin{equation}\label{prop_low_onl}
\begin{cases}
\alpha_{j}f(k)=\alpha_{j+1}\frac{f(k+1)^{k+1}}{f(k)^{k}} &\text{if }j\leq m-2\\
\alpha_{j}f(k)=\alpha_{j+1}\frac{g(h+1)^{h+1}}{g(h)^{h}} &\text {if }j=m-1\\
\alpha_{j}g(h)=\alpha_{j+1}\frac{g(h+1)^{h+1}}{g(h)^{h}} &\text {if }m\leq j\leq 2m-2\\
\alpha_{j}g(h)=\alpha_{j+1}g(1)& \text {if }j=2m-1
\end{cases},
\begin{cases}
\beta_{j} w_j=k,\ w_{j}=k^j &\text{if }j\leq m-1\\
\beta_{j} w_j=h,\ w_{j}=h^{j+1-m}k^{m-1} &\text {if }m\leq j\leq 2m-1\\
\beta_{j+1} w_j=1&\text{if }j\leq 2m-1
\end{cases}
\end{equation}\normalsize
Let $\sg$ be the strategy profile in which each client $j$ is assigned to  resource $r_j$. We show that $\sg$ is a state that can be possibly returned by the greedy algorithm when clients are processed in reverse order w.r.t. index $j$. We equivalently show that $\frac{{\sf NSW}(\sg^j)}{{\sf NSW}(\sg^{j+1})}\leq \frac{{\sf NSW}(\sg^{j+1},r_{j+1})}{{\sf NSW}(\sg^{j+1})}$ for any $j\leq 2m-1$, where $\sg^j$ denotes the partial assignment in which each client $t\geq j$ is assigned to resource $r_t$, and $(\sg^{j+1},r_{j+1})$ denotes the partial assignment in which each client $t\geq j+1$ is assigned to resource $r_t$ and client $j$ is assigned to resource $r_{j+1}$. 
%\begin{lemma}\label{lem_weig}
%By taking a sufficiently large $s$, the strategy profile $\sg$ is a pure Nash equilibrium . 
%\end{lemma}
%\begin{proof}
Let $j\in [2m-1]$. First of all, assume that $j\leq m-2$. By using (\ref{prop_low_onl}), we get 
\begin{align*}
&\frac{{\sf NSW}(\sg^j)}{{\sf NSW}(\sg^{j+1})}=\ell_{r_j}(k_{r_j}(\sg))^{k_{r_j}(\sg)}=\left(\alpha_j\hat{f}_j\left(\beta_j w_j\right)\right)^{w_j}=(\alpha_jf\left(k\right))^{w_j}=\left(\alpha_{j+1}\frac{f\left(k+1\right)^{k+1}}{f(k)^k}\right)^{w_j}\nonumber\\
&=\alpha_{j+1}^{w_j}\frac{f\left(k+1\right)^{kw_j+w_j}}{f(k)^{kw_j}}=\alpha_{j+1}^{w_j}\frac{f\left(k+1\right)^{w_{j+1}+w_j}}{f(k)^{w_{j+1}}}=\frac{\left(\alpha_{j+1}f\left(k+1\right)\right)^{w_{j+1}+w_j}}{\left(\alpha_{j+1}f(k)\right)^{w_{j+1}}}\nonumber\\
&=\frac{\left(\alpha_{j+1}f\left(\beta_{j+1}(w_{j+1}+w_j)\right)\right)^{w_{j+1}+w_j}}{\left(\alpha_{j+1}f(\beta_{j+1}w_{j+1})\right)^{w_{j+1}}}=\frac{\left(\alpha_{j+1}\hat{f}_{j+1}\left(\beta_{j+1}(w_{j+1}+w_j)\right)\right)^{w_{j+1}+w_j}}{\left(\alpha_{j+1}\hat{f}_{j+1}(\beta_{j+1}w_{j+1})\right)^{w_{j+1}}}\nonumber\\
&=\frac{\ell_{r_{j+1}}(k_{r_{j+1}}(\sg^{j+1},r_{j+1}))^{k_{r_{j+1}}(\sg^{j+1},r_{j+1})}}{\ell_{r_{j+1}}(k_{r_{j+1}}(\sg^{j+1}))^{k_{r_{j+1}}(\sg^{j+1})}}=\frac{{\sf NSW}(\sg^{j+1},r_{j+1})}{{\sf NSW}(\sg^{j+1})}.\nonumber
\end{align*}
The cases $j=m-1$, $m\leq j\leq 2m-2$, and $j=2m-1$ can be separately considered by exploiting (\ref{prop_low_onl}), so that one can analogously get
\begin{equation}\label{onl_low_ciao1}
\frac{{\sf NSW}(\sg^j)}{{\sf NSW}(\sg^{j+1})}=\left(\alpha_j\hat{f}_j\left(\beta_j w_j\right)\right)^{w_j}=\frac{\left(\alpha_{j+1}\hat{f}_{j+1}\left(\beta_{j+1}(w_{j+1}+w_j)\right)\right)^{w_{j+1}+w_j}}{\left(\alpha_{j+1}\hat{f}_{j+1}(\beta_{j+1}w_{j+1})\right)^{w_{j+1}}}=\frac{{\sf NSW}(\sg^{j+1},r_{j+1})}{{\sf NSW}(\sg^{j+1})},
\end{equation}
where we set $\left(\alpha_{2m}\hat{f}_{2m}(\beta_{2m}w_{2m})\right)^{w_{2m}}:=1$ and $w_{2m}:=0$. 
Now, let $\sg^*$ be the strategy profile of $\I(m)$ in which each client $j\in [m-1]$ is assigned to resource $r_{j+1}$. By exploiting the definitions of $\alpha_j$,$\beta_j$, $\hat{f}_j$, and $w_j$, and by considering a sufficiently large $m$, we have that:\small  
\begin{align}
&{\sf NPoA}(\I(m))\nonumber\\
\geq &\frac{{\sf NSW}(\sg)}{{\sf NSW}(\sg^*)}\nonumber\\
=&\left(\frac{\prod_{j=1}^{2m-1}\left(\alpha_j\hat{f}_j\left(\beta_j w_j\right)\right)^{w_j}}{\prod_{j=2}^{2m}\left(\alpha_{j}\hat{f}_j\left(\beta_j w_{j-1}\right)\right)^{w_{j-1}}}\right)^{\frac{1}{\sum_{j=1}^{2m-1}w_j}}\nonumber\\
=&\left(\frac{\prod_{j=1}^{2m-1}\left(\frac{\left(\alpha_{j+1}\hat{f}_{j+1}\left(\beta_{j+1}(w_{j+1}+w_j)\right)\right)^{w_{j+1}+w_j}}{\left(\alpha_{j+1}\hat{f}_{j+1}(\beta_{j+1}w_{j+1})\right)^{w_{j+1}}}\right)}{\prod_{j=2}^{2m}\left(\alpha_{j}\hat{f}_j\left(\beta_j w_{j-1}\right)\right)^{w_{j-1}}}\right)^{\frac{1}{\sum_{j=1}^{2m-1}w_j}}\label{onl_low_ciao2}\\
=&\left(\frac{\prod_{j=1}^{2m-1}\left(\frac{\left(\alpha_{j+1}\hat{f}_{j+1}\left(\beta_{j+1}(w_{j+1}+w_j)\right)\right)^{w_{j+1}+w_j}}{\left(\alpha_{j+1}\hat{f}_{j+1}(\beta_{j+1}w_{j+1})\right)^{w_{j+1}}}\right)}{\prod_{j=1}^{2m-1}\left(\alpha_{j+1}\hat{f}_{j+1}\left(\beta_{j+1} w_{j}\right)\right)^{w_{j}}}\right)^{\frac{1}{\sum_{j=1}^{2m-1}w_j}}\nonumber\\
=&\left(\prod_{j=1}^{2m-1}\left(\frac{\left(\alpha_{j+1}\hat{f}_{j+1}\left(\beta_{j+1}(w_{j+1}+w_j)\right)\right)^{w_{j+1}+w_j}}{\left(\alpha_{j+1}\hat{f}_{j+1}(\beta_{j+1}w_{j+1})\right)^{w_{j+1}}\left(\alpha_{j+1}\hat{f}_{j+1}\left(\beta_{j+1} w_{j}\right)\right)^{w_{j}}}\right)\right)^{\frac{1}{\sum_{j=1}^{2m-1}w_j}}\nonumber\\
=&\left(\prod_{j=1}^{m-2}\left(\frac{f(k+1)^{k^{j+1}+k^j}}{f(k)^{k^{j+1}}f(1)^{k^j}}\right) \prod_{j=m-1}^{2m-2}\left(\frac{g(h+1)^{h^{j+2-m}k^{m-1}+h^{j+1-m}k^{m-1}}}{g(h)^{h^{j+2-m}k^{m-1}}g(1)^{h^{j+1-m}k^{m-1}}}\right)\right)^{\frac{1}{\sum_{j=1}^{m-2}k^j+\sum_{j=m-1}^{2m-1}h^{j+1-m}k^{m-1}}}\nonumber\\
=&\left(\prod_{j=1}^{m-2}\left(\frac{f(k+1)^{k+1}}{f(k)^{k}f(1)}\right)^{k^j} \prod_{j=m-1}^{2m-2}\left(\frac{g(h+1)^{h+1}}{g(h)^{h}g(1)}\right)^{h^{j+1-m}k^{m-1}}\right)^{\frac{1}{\sum_{j=1}^{m-2}k^j+\sum_{j=m-1}^{2m-1}h^{j+1-m}k^{m-1}}}\nonumber\\
\geq &\left(\frac{f(k+1)^{k+1}}{f(k)^{k}f(1)}\right)^{\frac{1-h}{k-h}}\left(\frac{g(h+1)^{h+1}}{g(h)^{h}g(1)}\right)^{\frac{k-1}{k-h}}-\epsilon\label{onl_low_form_5}\\
> & M+\epsilon-\epsilon\label{onl_low_form_6}\\
= &M,\label{onl_low_form_7}
\end{align}\normalsize
where (\ref{onl_low_ciao2}) comes from (\ref{onl_low_ciao1}), (\ref{onl_low_form_5}) can be shown by using similar arguments as in the proof of Theorem \ref{thm_w_low} (see steps (\ref{w_low_form_5_app}) and (\ref{w_low_form_8_app})), and  (\ref{onl_low_form_6}) comes from (\ref{onl_low_form_0}). By (\ref{onl_low_form_7}), the claim follows. 

If $h=0$, we consider a load balancing instance defined as $\I(m)$, but restricted to resources $r_1,r_2,\ldots, r_m$ and to players in $[m-1]$. By using the same proof arguments as those used for $h>0$, one can show the claim as well. \qed
\end{proof}
\subsection{Proof of Corollary \ref{pol_onl_cor}}
The proof follows from the following lemma.
\begin{lemma}\label{lemma_CR_weighted_Polynomial latency functions}
$$\sup_{\substack{k_1\geq  o_1> 0,\\o_2>k_2\geq 0,\\f_1,f_2\in \mathcal{C}}}\left(\frac{f_1(k_1+o_1)^{k_1+o_1}}{f_1(k_1)^{k_1}f_1(o_1)^{o_1}}\right)^{\frac{o_2-k_2}{o_2 k_1-o_1k_2}}\left(\frac{f_2(k_2+o_2)^{k_2+o_2}}{f_2(k_2)^{k_2}f_2(o_2)^{o_2}}\right)^{\frac{k_1-o_1}{o_2 k_1-o_1k_2}}= 4^p.$$
%The Nash price of anarchy of weighted load balancing games with polynomial latency functions (even for symmetric games) of maximum degree $p$ is ${\sf NPoA}({\sf WLB}(\mathcal{P}(p)))={\sf NPoA}({\sf SWLB}(\mathcal{P}(p)))=2^p$. 
\end{lemma}
\begin{proof}
We have that
\begin{align}
&\sup_{k_1\geq  o_1> 0,o_2>k_2\geq 0,f_1,f_2\in \mathcal{C}}\left(\frac{f_1(k_1+o_1)^{k_1+o_1}}{f_1(k_1)^{k_1}f_1(o_1)^{o_1}}\right)^{\frac{o_2-k_2}{o_2 k_1-o_1k_2}}\left(\frac{f_2(k_2+o_2)^{k_2+o_2}}{f_2(k_2)^{k_2}f_2(o_2)^{o_2}}\right)^{\frac{k_1-o_1}{o_2 k_1-o_1k_2}}\nonumber\\
=&\sup_{\substack{k_1\geq o_1> 0,\\o_2>k_2\geq 0,\\\alpha_0,\ldots, \alpha_p\geq 0,\\\beta_0,\ldots,\beta_p\geq 0}}\left(\frac{\left(\sum_{d=0}^{p}\alpha_d(k_1+o_1)^d\right)^{k_1+o_1}}{\left(\sum_{d=0}^{p}\alpha_dk_1^d\right)^{k_1}\left(\sum_{d=0}^{p}\alpha_do_1^d\right)^{o_1}}\right)^{\frac{o_2-k_2}{o_2 k_1-o_1k_2}}\left(\frac{\left(\sum_{d=0}^{p}\beta_d(k_2+o_2)^d\right)^{k_2+o_2}}{\left(\sum_{d=0}^{p}\beta_dk_2^d\right)^{k_2}\left(\sum_{d=0}^{p}\beta_do_2^d\right)^{o_2}}\right)^{\frac{k_1-o_1}{o_2 k_1-o_1k_2}}\nonumber\\
=&\sup_{\substack{k_1\geq o_1> 0,o_2>k_2\geq 0,\\\alpha_0,\ldots, \alpha_p,\beta_0,\ldots,\beta_p\geq 0}}\left(\left(\frac{\sum_{d=0}^{p}\alpha_d(k_1+o_1)^d}{\sum_{d=0}^{p}\alpha_dk_1^d}\right)^{k_1}\left(\frac{\sum_{d=0}^{p}\alpha_d(k_1+o_1)^d}{\sum_{d=0}^{p}\alpha_do_1^d}\right)^{o_1}\right)^{\frac{o_2-k_2}{o_2 k_1-o_1k_2}}\nonumber\\
&\quad\quad\quad\quad\quad\quad\quad\cdot \left(\left(\frac{\sum_{d=0}^{p}\beta_d(k_2+o_2)^d}{\sum_{d=0}^{p}\beta_dk_2^d}\right)^{k_2}\left(\frac{\sum_{d=0}^{p}\beta_d(k_2+o_2)^d}{\sum_{d=0}^{p}\beta_d o_2^d}\right)^{o_2}\right)^{\frac{k_1-o_1}{o_2 k_1-o_1k_2}}\nonumber\\
=&\sup_{k_1\geq o_1> 0,o_2>k_2\geq 0}\left(\left(\max_{d\in [p]\cup\{0\}}\frac{(k_1+o_1)^d}{k_1^d}\right)^{k_1}\left(\max_{d\in [p]\cup\{0\}}\frac{(k_1+o_1)^d}{o_1^d}\right)^{o_1}\right)^{\frac{o_2-k_2}{o_2 k_1-o_1k_2}}\nonumber\\
&\quad\quad\quad\quad\quad\quad\quad\cdot \left(\left(\max_{d\in [p]\cup\{0\}}\frac{(k_2+o_2)^d}{k_2^d}\right)^{k_2}\left(\max_{d\in [p]\cup\{0\}}\frac{(k_2+o_2)^d}{o_2^d}\right)^{o_2}\right)^{\frac{k_1-o_1}{o_2 k_1-o_1k_2}}\nonumber\\
=&\sup_{\substack{k_1\geq o_1> 0,\\o_2>k_2\geq 0}}\left(\left(\frac{(k_1+o_1)^p}{k_1^p}\right)^{k_1}\left(\frac{(k_1+o_1)^p}{o_1^p}\right)^{o_1}\right)^{\frac{o_2-k_2}{o_2 k_1-o_1k_2}}\left(\left(\frac{(k_2+o_2)^p}{k_2^p}\right)^{k_2}\left(\frac{(k_2+o_2)^p}{o_2^p}\right)^{o_2}\right)^{\frac{k_1-o_1}{o_2 k_1-o_1k_2}}\nonumber\\
=&\sup_{k\geq 1,0\leq h<1}\left(\left(\frac{(k+1)^{k+1}}{k^k}\right)^{\frac{1-h}{k-h}}\left(\frac{(h+1)^{h+1}}{h^h}\right)^{\frac{k-1}{k-h}}\right)^p,\label{onl_pol_form_11}
\end{align}
where (\ref{onl_pol_form_11}) can be obtained  by setting $k:=k_1/o_1$ and $h:=k_2/o_2$. Now, we show that the maximum value of function $F(k,h):=\left(\frac{(k+1)^{k+1}}{k^k}\right)^{\frac{1-h}{k-h}}\left(\frac{(h+1)^{h+1}}{h^h}\right)^{\frac{k-1}{k-h}}$ over $k\geq 1$ and $0\leq h<1$ is equal to $4$. Observe that $
\ln(F(k,h))=\frac{1-h}{k-h}((k+1)\ln(k+1)-k\ln(k))+\frac{k-1}{k-h}((h+1)\ln(h+1)-h\ln(h))\leq \left(\frac{1-h}{k-h}(k+1)+\frac{k-1}{k-h}(h+1)\right)\ln\left(\frac{1-h}{k-h}(k+1)+\frac{k-1}{k-h}(h+1)\right)$, 
where the second last inequality holds because of the concavity of the function $g$ defined as $g(x):=(x+1)\ln(x+1)-x\ln(x)$ and since $\ln(F(k,h))$ is defined as convex combination of $g(k)$ and $g(h)$. Thus, we get
\begin{align}
F(k,h)&\leq \left(\frac{1-h}{k-h}(k+1)+\frac{k-1}{k-h}(h+1)\right)^{\frac{1-h}{k-h}(k+1)+\frac{k-1}{k-h}(h+1)}\nonumber\\
&=\left(\frac{(k-h)+(k-h)}{k-h}\right)^{\frac{(k-h)+(k-h)}{k-h}}=2^2=4.\label{onl_pol_form_13}
\end{align}
Finally, since $F(k,h)=4$ for $k=1$ and $h=0$, and because of (\ref{onl_pol_form_13}), we have that the maximum of $F(k,h)$ over $k\geq 1$ and $0\leq h<1$ is $4$. Thus, we get that (\ref{onl_pol_form_11}) is at most $4^p$. \qed
\end{proof}
\subsection{Tightness of the Upper Bound of Corollary \ref{pol_onl_cor} w.r.t. any Online Algorithm.}
\begin{theorem}\label{thm_onl_low_problem}
The competitive ratio of any online algorithm ${\sf A}$ applied to load balancing instances with polynomial latencies of maximum degree $p$ is at least ${\sf CR}_{\sf A}(\mathcal{P}(p))\geq 4^p$, even for instances with identical resources.
\end{theorem}
\begin{proof}
We equivalently show that, for any online algorithm $A$ and $\epsilon>0$, there exists a load balancing instance $\I$ such that ${\sf \CR}_{\sf A}(\I)\geq 4^p-\epsilon$. We construct an instance similar to that defined in Theorem 17 of \cite{CFKKM11}. Given an integer $m\geq 0$ and a real number $w>0$, let $\I(m)$ be a load balancing instance with identical polynomial latency functions of type $\ell(x)=x^p$, and recursively defined as follows: 
\begin{itemize}
\item If $m=0$, $\I(m)$ has no clients and there is a unique resource denoted as {\em fundamental resource} of $\I(0)$. 
\item If $m\geq 1$, then: (i) $\I(m)$ contains a sub-instance equivalent to $\I(i-1)$ for any $i\in [m]$; (ii) $\I(m)$ has a further resource $r$ denoted as {\em fundamental resource} of $\I(m)$; (iii) there are further $m$ clients such that, for any $i\in [m]$, the $i$-th client has weight $w_i:=2^{i-1}$ and can select among $r$ and the fundamental resource $r(i)$ of the sub-instance of type $\I(i-1)$ included in $\I(m)$; (iv) for any client $i\in [m]$, $r$ and $r(i)$ are respectively denoted as first and second resource of the $i$-th client included in $\I(m)$.
\end{itemize}
Let $\sg$ and $\sg^*$ be the states of $\I(m)$ in which each client is assigned to her first and second resource, respectively. We have  that $\sg$ is a state that can be returned by any online algorithm if clients are processed according to the following partial ordering: (i) given two clients $i_1$ and $i_2$ having their first resource in sub-instances of type $\I(m_1)$ and $\I(m_2)$ respectively, if $m_1<m_2$ then client $i_1$ is processed before client $i_2$; (ii) the clients defined in the same sub-instance are processed in increasing order with respect to their weights. This fact is true since each time the greedy algorithm processes some client $i$ according to the partial ordering defined above, the congestions of the first and the second resource of that client are equal. Thus, since the latency functions are equal too, any online algorithm cannot distinguish between the two resources selectable by each client, and by symmetry both choices can potentially lead to the same worst-case competitive ratio. 

We have the following fact:
\begin{fact}\label{onl_lem_sub}
Given two integers $m\geq 1$ and $i\in [m-1]\cup\{0\}$ such that $j\geq i$, the number $N(m,i)$ of sub-instances of $\I(m)$  equivalent to $\I(j)$ for some $j\geq i$ is $N(m,i)=2^{m-i}$.
\end{fact}
\begin{proof}
We show the claim by induction on $h(i):=m-i\geq 0$. If $h(i)=0$ the unique sub-instance equivalent to $\I(j)$ for some $j\geq i$ is the entire instance $\I(m)$, thus $N(m,i)=1=2^{h(i)}=2^{m-i}$ and the base step holds. Now, assume that the claim holds for any $h(i)\geq 0$. Observe that we can associate in a one-to-one correspondence each sub-instance that is equivalent to $\I(j)$ for some $j\geq i$, with a sub-instance equivalent to $\I(i-1)$, that is $N(m,i)=N(m,i-1)-N(m,i)\Rightarrow N(m,i-1)=2N(m,i)$. Thus, we have that $N(m,i-1)=2N(m,i)=2\cdot 2^{h(i)}=2^{m-i+1}=2^{h(i)+1}$, and the inductive step holds. \qed
\end{proof}
Let $N(m,i)$ be defined as in Fact \ref{onl_lem_sub} and let $R(i)$ be the set of fundamental resources for sub-instances of type $\I(i)$. Observe that, for any $i\in [m]$ and resource $r$ such that $i$ clients select $r$ as first resource, $r$ is the fundamental resource of a sub-instance of type $\I(i)$, i.e., $r\in R(i)$. Thus, by exploiting Fact \ref{onl_lem_sub}, we get 
\begin{align}
{\sf NSW}(\sg)&=\left(\prod_{i\in [m]}\prod_{r\in R(i)}\ell(k_r(\sg))^{k_r(\sg)}\right)^{\frac{1}{\sum_{r\in R}k_r(\sg)}}\nonumber\\
&=\left(\prod_{i\in [m]}\ell\left(\sum_{j=1}^iw_j\right)^{\left(\sum_{j=1}^iw_j\right)|R(i)|}\right)^{\frac{1}{\sum_{i\in [m]}\left(\sum_{j=1}^iw_j\right)|R(i)|}}\nonumber\\
&=\left(\prod_{i\in [m]}\ell\left(\sum_{j=1}^i2^{j-1}\right)^{\left(\sum_{j=1}^i2^{j-1}\right)(N(m,i)-N(m,i+1))}\right)^{\frac{1}{\sum_{i\in [m]}\left(\sum_{j=1}^i2^{j-1}\right)(N(m,i)-N(m,i+1))}}\nonumber\\
&=\left(\prod_{i\in [m]}\left(2^i-1\right)^{p\left(2^i-1\right)2^{m-i-1}}\right)^{\frac{1}{\sum_{i\in [m]}\left(2^i-1\right)2^{m-i-1}}}\label{onl_part_alg}
\end{align}
and 
\begin{align}
{\sf NSW}(\sg^*)&=\left(\prod_{i\in [m]}\prod_{r\in R(i)}\prod_{j\in [i]}\ell(k_{r(j)}(\sg^*))^{k_{r(j)}(\sg^*)}\right)^{\frac{1}{\sum_{r\in R}k_r(\sg)}}\nonumber\\
&=\left(\prod_{i\in [m]}\prod_{j\in [i]}\ell\left(w_j\right)^{w_j(N(m,i)-N(m,i+1))}\right)^{\frac{1}{\sum_{i\in [m]}\left(2^i-1\right)2^{m-i-1}}}\nonumber\\
&=\left(\prod_{i\in [m]}\prod_{j\in [i]}\left(2^{j-1}\right)^{p2^{j-1}2^{m-i-1}}\right)^{\frac{1}{\sum_{i\in [m]}\left(2^i-1\right)2^{m-i-1}}}\nonumber\\
&=\left(\prod_{i\in [m]}2^{p\left(\sum_{j=0}^{i-1}j2^{j}\right)2^{m-i-1}}\right)^{\frac{1}{\sum_{i\in [m]}\left(2^i-1\right)2^{m-i-1}}}\nonumber\\
&=\left(\prod_{i\in [m]}2^{p\left(i2^i-2\left(2^i-1\right)\right)2^{m-i-1}}\right)^{\frac{1}{\sum_{i\in [m]}\left(2^i-1\right)2^{m-i-1}}}\label{onl_part_opt}
\end{align}
Let $\epsilon>0$. By (\ref{onl_part_alg}) and (\ref{onl_part_opt}), and by taking a sufficiently large integer $m>1$, we get
\begin{align}
{\sf CR}_{\sf A}(\I)&\geq \frac{{\sf NSW}(\sg)}{{\sf NSW}(\sg^*)}\nonumber\\
&= \left(\frac{\prod_{i\in [m]}\left(2^i-1\right)^{p\left(2^i-1\right)2^{m-i-1}}}{\prod_{i\in [m]}2^{p\left(i2^i-2\left(2^i-1\right)\right)2^{m-i-1}}}\right)^{\frac{1}{\sum_{i\in [m]}\left(2^i-1\right)2^{m-i-1}}}\nonumber\\
&= \left(\frac{\prod_{i\in [m]}\left(2^i-1\right)^{p\left(2^i-1\right)2^{-i-1}}}{\prod_{i\in [m]}2^{p\left(i2^i-2\left(2^i-1\right)\right)2^{-i-1}}}\right)^{\frac{1}{\sum_{i\in [m]}\left(2^i-1\right)2^{-i-1}}}\nonumber\\
&=  \left(\frac{\prod_{i\in [m]}\left(2^i\right)^{p\left(2^i-1\right)2^{-i-1}}}{\prod_{i\in [m]}2^{p\left(i2^i-2\left(2^i-1\right)\right)2^{-i-1}}}\right)^{\frac{1}{\sum_{i\in [m]}\left(2^i-1\right)2^{-i-1}}}\prod_{i\in [m]}\left(\frac{2^i-1}{2^i}\right)^{\frac{p\left(2^i-1\right)2^{-i-1}}{\sum_{i\in [m]}\left(2^i-1\right)2^{-i-1}}}\nonumber\\
&=  \left(2^p\right)^{\frac{\sum_{i\in [m]}\left(-i2^{-i-1}+1-2^{-i}\right)}{
\sum_{i\in [m]}(1/2-2^{-i-1})}}\prod_{i\in [m]}\left(\frac{2^i-1}{2^i}\right)^{\frac{p\left(1/2-2^{-i-1}\right)}{\sum_{i\in [m]}\left(1/2-2^{-i-1}\right)}}\label{onl_pol_fin_0}
\end{align}
We have the following fact:
\begin{fact}\label{fact_onl}
\begin{equation*}
\lim_{m\rightarrow \infty}\prod_{i\in [m]}\left(\frac{2^i-1}{2^i}\right)^{\frac{p\left(1/2-2^{-i-1}\right)}{\sum_{i\in [m]}\left(1/2-2^{-i-1}\right)}}=1.
\end{equation*}
\end{fact}
\begin{proof}
Set $\alpha_i:=p\ln\left(\frac{2^i-1}{2^i}\right)$ and $\beta_i:=\left(1/2-2^{-i-1}\right)$. We will equivalently show that $\lim_{m\rightarrow \infty}\frac{\sum_{i=1}^m\alpha_i\beta_i}{\sum_{i=1}^m\beta_i}=0$, since, by exponentiating this equality, we get the claim. Set $a_m:=\sum_{i=1}^m\alpha_i\beta_i$ and $b_m:=\sum_{i=1}^m\beta_i$. We have that sequence $(b_m)_{m\geq 1}$ is positive, increasing, and unbounded. Thus, by the Stolz-Cesaro Theorem, we have that $\lim_{m\rightarrow\infty}\frac{a_m}{b_m}=\lim_{m\rightarrow\infty}\frac{a_{m+1}-a_m}{b_{m+1}-b_m}$. We conclude that $\lim_{m\rightarrow \infty}\frac{\sum_{i=1}^m\alpha_i\beta_i}{\sum_{i=1}^m\beta_i}=\lim_{m\rightarrow\infty}\frac{a_m}{b_m}=\lim_{m\rightarrow\infty}\frac{a_{m+1}-a_m}{b_{m+1}-b_m}=\lim_{m\rightarrow \infty}\frac{\alpha_m\beta_m}{\beta_m}=\lim_{m\rightarrow \infty}p\ln\left(\frac{2^m-1}{2^m}\right)=0$, and the claim follows. \qed
\end{proof}
By continuing from (\ref{onl_pol_fin_0}), we get
\begin{align}
&=  \left(2^p\right)^{\frac{\sum_{i\in [m]}\left(-i2^{-i-1}+1-2^{-i}\right)}{
\sum_{i\in [m]}(1/2-2^{-i-1})}}\prod_{i\in [m]}\left(\frac{2^i-1}{2^i}\right)^{\frac{p\left(1/2-2^{-i-1}\right)}{\sum_{i\in [m]}\left(1/2-2^{-i-1}\right)}}\nonumber\\
&\geq \lim_{m\rightarrow \infty }\left(2^p\right)^{\frac{\sum_{i\in [m]}\left(-i2^{-i-1}+1-2^{-i}\right)}{
\sum_{i\in [m]}(1/2-2^{-i-1})}}\prod_{i\in [m]}\left(\frac{2^i-1}{2^i}\right)^{\frac{p\left(1/2-2^{-i-1}\right)}{\sum_{i\in [m]}\left(1/2-2^{-i-1}\right)}}-\epsilon \nonumber\\
&=\lim_{m\rightarrow \infty}\left(2^p\right)^{\frac{\sum_{i\in [m]}\left(-i2^{-i-1}+1-2^{-i}\right)}{
\sum_{i\in [m]}(1/2-2^{-i-1})}}-\epsilon\label{app_fact_onl}\\
&=\lim_{m\rightarrow \infty}\left(2^p\right)^{\frac{2^{-m-1}m+m+2^{1 - m}-2}{
1/2\left(m+2^{-m}-1\right)}}-\epsilon\nonumber\\
&=\left(2^p\right)^{\left(\lim_{m\rightarrow \infty}\frac{2^{-m-1}m+m+2^{1 - m}-2}{
1/2\left(m+2^{-m}-1\right)}\right)}-\epsilon\nonumber\\
&=\left(2^p\right)^{\left(\lim_{m\rightarrow \infty}\frac{m}{
1/2(m) }\right)}-\epsilon\nonumber\\
&=\left(2^p\right)^{2}-\epsilon\nonumber\\
&=4^p-\epsilon,\nonumber
\end{align}
where (\ref{app_fact_onl}) comes from Fact \ref{fact_onl}.  We conclude that there exists a load balancing instance $\I$ such that ${\sf CR}_{\sf A}(\I)\geq 4^p-\epsilon$, thus, for the arbitrariness of $\epsilon$, the claim follows.
\section{Lower bound for Linear Congestion Games}
Unweighted congestion games are a further generalization of unweighted load balancing games. The difference is that the strategy set of each player $i\in \N$ is a collection ${\Sigma}_i\subseteq 2^R\setminus\{\emptyset\}$, i.e., a strategy is a non-empty subset of $R$. Furthermore, given a strategy profile $\sg=(\sigma_1,\ldots, \sigma_n)$ (with $\sigma_i\in {\Sigma}_i$), the cost of each player $i\in\N$ is $cost_i(\sg):=\sum_{j\in \sigma_i}\ell_j(k_j(\sg))$, where $k_j(\sg):=|i\in \N:j\in \sigma_i|$ is the congestion of resource $j$ in strategy profile $\sg$. In the following theorem, we show that, even for linear latency functions, the Nash price of anarchy of unweighted congestion games with linear latency functions is non-constant in the number of players, differently from the case of load balancing games. This fact exhibits a substantial difference with respect to the case of the price of anarchy when the considered social function is the sum of the players' costs. Indeed, in such case, the price of anarchy for linear congestion games is finite, and the price of anarchy of load balancing games is as high as that of general linear congestion games. 
\begin{theorem}\label{thm_CG}
The Nash price of anarchy of linear congestion games is at least $n^{1-o(1)}$, where $n$ is the number of players (and $o(1)$ is an infinitesimal w.r.t. to $n$). 
\end{theorem}
\begin{proof}
We show that, for any $\epsilon\in (0,1/2)$, there exists a congestion game ${\sf CG}$ with linear latency functions and $n\geq 2$ players such that:
\begin{equation}\label{unw_gen_low_form_0}
{\sf NPoA}({\sf CG})\geq \lceil n\epsilon\rceil^{1-\frac{\lceil n\epsilon\rceil}{n}},
\end{equation}
and this fact will imply the claim, as $\lceil n\epsilon\rceil^{1-\frac{\lceil n\epsilon\rceil}{n}}\in \Theta(n^{1-\epsilon})$ for any fixed $\epsilon\in (0,1/2)$. 
Let $\epsilon\in (0,1/2)$, $n\geq 2$, and $m:=\lceil n\epsilon\rceil$. Let ${\sf CG}(n,\epsilon)$ be an unweighted congestion game with $n$ players defined as follows: The set of resources is organized into three groups $R_1,R_2,R_3$, with $R_j:=\{r_{j,1},\ldots, r_{j,{n-m}}\}$ for any $j\in [2]$, and $R_3:=\{r_{3,1},\ldots, r_{3,{m}}\}$. The latency function of each resource $r_{j,h}$ is $\ell_{r_{j,h}}(x):=\alpha_{j}x$, where $\alpha_1=m+1$, $\alpha_2=1$, and $\alpha_3=m$. There are two groups of players $\N_1,\N_2$, with $\N_1:=\{i_{1,1},\ldots, i_{1,{n-m}}\}$ and $\N_2:=\{i_{2,1},\ldots, i_{2,{m}}\}$. 
Each player $i_{1,h}\in N_1$ has two strategies $S_{1,h}$ and $S_{1,h}^*$ defined as $S_{1,h}:=\{r_{1,h}\}$ and $S_{1,h}^*:=\{r_{2,h}\}$, and each player $i_{2,h}\in N_2$ has two strategies $S_{2,h}$ and $S_{2,h}^*$ defined as $S_{2,h}:=R_2$ and $S_{3,h}^*:=\{r_{3,h}\}$. 
Let $\sg$ (resp. $\sg^*$) be the strategy profile such that each player $i_{t,h}$ plays strategy $S_{t,h}$ (resp. $S_{t,h}^*$), for any $t\in [2]$. One can easily show that $cost_i(\sg)=cost_i(\sg_{-i},\sigma_i^*)$ for any player $i$, thus $\sg$ is a pure Nash equilibrium. We have that:
\begin{align}
&\ {\sf NPoA}({\sf CG}(n,\epsilon))\nonumber\\
&\geq \frac{{\sf NSW}(\sg)}{{\sf NSW}(\sg^*)}\nonumber\\
&=\left(\left(\prod_{i\in \N_1}\frac{cost_i(\sg)}{cost_i(\sg^*)}\right)\left(\prod_{i\in \N_2}\frac{cost_i(\sg)}{cost_i(\sg^*)}\right)\right)^{\frac{1}{n}}\nonumber\\
&=\left(\left(\prod_{i\in \N_1}\frac{cost_i(\sg_{-i},\sigma_i^*)}{cost_i(\sg^*)}\right)\left(\prod_{i\in \N_2}\frac{cost_i(\sg_{-i},\sigma_i^*)}{cost_i(\sg^*)}\right)\right)^{\frac{1}{n}}\nonumber\\
&=\left(\left(\frac{\alpha_2(m+1)}{\alpha_2}\right)^{n-m}\left(\frac{\alpha_3}{\alpha_3}\right)^m\right)^{\frac{1}{n}}\nonumber\\
&=\left(m+1\right)^{\frac{n-m}{n}}\nonumber\\
&\geq \lceil n\epsilon\rceil^{1-\frac{\lceil n\epsilon\rceil}{n}},
\end{align}
thus (\ref{unw_gen_low_form_0}) holds, and the claim follows.\qed
\end{proof}
\end{proof}
\end{document}